\documentclass[11pt,letterpaper]{article}

\usepackage{amsmath,amsthm,amsfonts,amssymb}
\usepackage{thm-restate}
\usepackage{fullpage}
\usepackage[utf8]{inputenc}
\usepackage[dvipsnames]{xcolor}
\usepackage{xspace,enumitem}
\usepackage[hypertexnames=false,colorlinks=true,urlcolor=Blue,citecolor=Blue,linkcolor=BrickRed]{hyperref}
\usepackage[capitalise]{cleveref}
\usepackage[OT4]{fontenc}
\usepackage{ifpdf}
\usepackage{todonotes}
\usepackage{thmtools}
\usepackage{algorithm}
\usepackage[noend]{algpseudocode}
\usepackage{algorithmicx}
\usepackage{makecell}
\usepackage{multirow}
\usetikzlibrary{patterns}

\usepackage[margin=1in]{geometry}
\title{Fully Dynamic Strongly Connected Components in Planar Digraphs}

\newcommand{\email}[1]{\href{mailto:#1}{#1}}
\date{\vspace{-5ex}}
\author{Adam Karczmarz\thanks{University of Warsaw and IDEAS NCBR, Poland. \email{a.karczmarz@mimuw.edu.pl}. Partially supported by the ERC CoG grant TUgbOAT no 772346 and the National Science Centre (NCN) grant no. 2022/47/D/ST6/02184.}
\and Marcin Smulewicz\thanks{University of Warsaw, Poland. \email{m.smulewicz@mimuw.edu.pl}}}

\theoremstyle{plain}
\newtheorem{theorem}{Theorem}[section]
\newtheorem{lemma}[theorem]{Lemma}
\newtheorem{corollary}[theorem]{Corollary}

\newtheorem{definition}[theorem]{Definition}

\newtheorem{remark}[theorem]{Remark}

\def\polylog{\operatorname{polylog}}

\newcommand{\Ot}{\ensuremath{\widetilde{O}}}
\newcommand{\eps}{\ensuremath{\epsilon}}
\newcommand{\dist}{\text{dist}}

\newcommand{\Vor}{\text{Vor}}

\newcommand{\rdiv}{\mathcal{R}}
\newcommand{\bnd}{\partial}

\newcommand{\sccs}{\mathcal{S}}
\newcommand{\rev}[1]{{#1}^{\textrm{R}}}

\begin{document}

\maketitle

  \begin{abstract}
    In this paper, we consider maintaining strongly connected components (SCCs) of
    a directed planar graph subject to edge insertions and deletions.
    We show a data structure maintaining an implicit representation of the SCCs
    within $\Ot(n^{6/7})$ worst-case time per update.
    The data structure supports, in $O(\log^2{n})$ time, reporting
    vertices of any specified SCC (with constant overhead per reported vertex) and
    aggregating vertex information (e.g., computing the maximum label) over all the vertices of that SCC.
    Furthermore, it can maintain global information about the structure of SCCs,
    such as the number of SCCs, or the size of the largest SCC.
    
    To the best of our knowledge, no fully dynamic SCCs data structures with sublinear
    update time have been previously known for any major subclass of digraphs.
    Our result should be contrasted with the $n^{1-o(1)}$ amortized update time lower bound conditional on SETH~\cite{AbboudW14},~which holds even for dynamically maintaining whether a general digraph has more than two SCCs.
\end{abstract}

\section{Introduction}

Two vertices of a directed graph $G=(V,E)$ are called strongly connected if they
can reach each other using paths in $G$. Pairwise strong connectivity is an equivalence
relation and the strongly connected components (SCCs) of $G$ are its equivalence classes.
Computing the SCCs is among the most classical and fundamental algorithmic problems
on digraphs and there exists a number of linear-time algorithms for
that~\cite{Gabow00,Sharir,Tarjan72}.
Therefore, it is no surprise that \emph{maintaining} SCCs has been one
of the most actively studied problems on \emph{dynamic directed graphs}~\cite{AbboudW14, BenderFGT16, BernsteinDP21, BernsteinGS20, BernsteinPW19, ChechikHILP16, chen2023almostlinear, HaeuplerKMST12, ItalianoKLS17, KarczmarzS23, Lacki13, RodittyZ08}.

When maintaining the strongly connected components, the information we care about may vary.
First, we could be interested in efficiently answering \emph{pairwise} strong connectivity
queries: given $u,v\in V$, decide whether $u$ and $v$ are strongly connected.
Pairwise strong connectivity queries, however, cannot easily provide any information about
the \emph{global} structure of SCCs (such as the number of SCCs, the size of the largest
SCC). Neither they enable, e.g., listing the vertices strongly connected to some $u\in V$.
This is why, in the following, we distinguish between \emph{dynamic pairwise strong connectivity} and \emph{dynamic SCCs} data structures which provide a more global view.
In particular, all the data about the SCCs can be easily accessed if the SCCs are maintained \emph{explicitly},
e.g., if the SCC identifier of every vertex is stored at all times and explicitly updated.
\subsection{Previous work}
In the following, let $n=|V|$ and $m=|E|$.
Dynamic graph data structures are traditionally studies in \emph{incremental}, \emph{decremental}
or \emph{fully dynamic} settings, which permit the graph to evolve by either only edge insertions,
only deletions, or both, respectively.
A decremental data structure maintaining SCCs with near-optimal total
update time is known~\cite{BernsteinPW19}. Very recently, a deterministic data structure with $m^{1+o(1)}$ total update
time has been obtained also for the incremental setting~\cite{chen2023almostlinear}.
Both these state-of-the art data structures maintain the SCCs explicitly.

The \emph{fully dynamic} variant -- which is our focus in this paper -- although
the most natural, has been studied the least.
First of all, there is strong evidence that a non-trivial dynamic SCCs data structure
for sparse graphs cannot exist.
If the SCCs have to be maintained explicitly,
then a single update can cause a rather dramatic $\Omega(n)$-sized amortized change in the set of SCCs\footnote{Consider a directed cycle and switching its arbitrary single edge on and off.}.
As a result, an explicit update may be asymptotically as costly as recomputing SCCs from scratch.
This argument -- applicable also for maintaining \emph{connected components} of an undirected graph -- does not exclude the possibility of maintaining an \emph{implicit} representation
of the SCCs in sublinear time, though.
After all, there exist very efficient fully dynamic connectivity data structures, e.g.,~\cite{HolmLT01, HuangHKPT23, Wulff-Nilsen13a}, typically maintaining also an explicit spanning forest which allows retrieving any ``global'' component-wise information one can think of rather easily.
However,~\cite{AbboudW14}~showed that even for maintaining a single-bit
information whether $G$ has \emph{more than two SCCs}, a data
structure with $O(n^{1-\eps})$ amortized time
is unlikely, as it would break the Orthogonal Vectors conjecture implied by the
SETH~\cite{ImpagliazzoP01, Williams05}.\footnote{In~\cite{AbboudW14}, a conditional lower bound of the same strength is also derived
for the dynamic \#SSR problem where the goal is to dynamically count vertices reachable from a source $s\in V$.}
This considerably limits the possible global information about the SCCs that can be
maintained within sublinear time per update.

For denser graphs,~\cite{AbboudW14} also proved that maintaining essentially any (even pairwise) information about SCCs dynamically
within truly subquadratic update time has to rely on fast matrix multiplication.
And indeed, that pairwise strong connectivity can be maintained this way follows
easily from the dynamic matrix inverse-based dynamic $st$-reachability data structures~\cite{BrandNS19, Sankowski05}.
More recently, \cite{KarczmarzS23} showed that in fact SCCs can be maintained
explicitly in $O(n^{1.529})$ worst-case time per update.
They also proved that maintaining whether $G$ has just a single SCC (\emph{dynamic~SC}) is easier\footnote{Interestingly, the SETH-based lower bound
of~\cite{AbboudW14} does not apply to the dynamic SC problem.}
and can be achieved within $O(n^{1.406})$ worst-case time per update.
Both these bounds are tight conditional on the appropriate variants~\cite{BrandNS19} of the OMv conjecture~\cite{HenzingerKNS15}.

In summary, the complexity of maintaining SCCs in general
directed graphs is rather well-understood now. In partially dynamic
settings, the known bounds are near optimal unconditionally,
whereas in the fully dynamic setting, the picture appears complete
unless some popular hardness conjectures are proven wrong.
In particular, for general sparse digraphs, no (asymptotically) non-trivial fully
dynamic SCCs data structure can exist.

\paragraph{Planar graphs.}
It is thus natural to ask whether 
non-trivial dynamic SCCs data structures are possible if we limit our attention
to some significant class of sparse digraphs.
And indeed, this question has been partially addressed for \emph{planar digraphs} in the past.
Since pairwise $s,t$-strong connectivity queries reduce to two $s,t$-reachability
queries, the known planar dynamic reachability data structures~\cite{DiksS07, Subramanian93} imply that
sublinear ($\Ot(n^{2/3})$ or $\Ot(\sqrt{n})$-time, depending on whether embedding-respecting insertions are required)
updates/queries are possible for \emph{pairwise} strong connectivity.
Another trade-off for dynamic pairwise strong connectivity has been showed by~\cite{Charalampopoulos22}.
Namely,
they showed a fully dynamic data structure for planar graphs
with $\Ot(n^{4/5})$ worst-case update time that can produce an identifier $s_v$
of an SCC of a given query vertex $v$ in $O(\log^2{n})$ time. Whereas this is
slightly more general\footnote{Than answering pairwise strong connectivity queries.
Using the SCC-identifiers, one can, e.g., partition any $k$ vertices of $G$ into strongly connected classes in $\Ot(k)$~time, whereas using pairwise queries this requires $\Theta(k^2)$ queries.}, it still not powerful enough to enable efficiently maintaining any of the global
data about the SCCs of a dynamic planar digraph
such as the SCCs count.

To the best of our knowledge, the question whether a more robust -- that is, giving a more ``global'' perspective on the SCCs beyond only supporting pairwise queries -- fully dynamic SCCs data structure for planar digraphs (or digraphs from any other interesting class) with sublinear update time is possible has not been addressed before.

\subsection{Our results}
In this paper, we address the posed question in the case of planar directed
graphs. Specifically, our main result is a dynamic SCCs data structure summarized by the following theorem.
  \begin{restatable}{theorem}{planarsccs}\label{t:planar-sccs}
    Let $G$ be a planar digraph subject to planarity-preserving edge insertions and deletions.
    There exists a data structure maintaining the strongly connected components of $G$ implicitly in $\Ot(n^{6/7})$ worst-case time per update.
    Specifically:
    \begin{itemize}
      \item The data structure maintains the number of SCCs and the size of the largest SCC in $G$.
      \item For any query vertex $v$, in $O(\log^2{n})$ time the data structure can compute the size of the SCC
    of $v$, and enable reporting the elements of the SCC of $v$ in $O(1)$ worst-case time per element.
    \end{itemize}
  \end{restatable}
  In particular, Theorem~\ref{t:planar-sccs} constitutes the first known fully dynamic SCCs data structure with sublinear update time for any significant class of sparse digraphs.
It also shows that the conditional lower bound of~\cite{AbboudW14} does not hold in planar digraphs.
  
The data structure of Theorem~\ref{t:planar-sccs} is deterministic and
does not require the edge insertions to respect any fixed embedding of the graph
(this also applies to side results discussed below).
Obtaining more efficient data structures for fully dynamic embedding-respecting updates
is an interesting direction (see, e.g.,~\cite{DiksS07}) that is beyond the scope of this paper.

\paragraph{Related problems.}
  Motivated by the discrepancies between the known bounds for dynamic SCCs and dynamic SC in general digraphs (both from the lower-~\cite{AbboudW14} and upper bounds~\cite{KarczmarzS23} perspective), we also complement Theorem~\ref{t:planar-sccs} with a significantly simpler and faster data structure suggesting that the dynamic SC might be easier (than dynamic SCCs) in planar digraphs as well.\footnote{Clearly, one could use Theorem~\ref{t:planar-sccs} for dynamic SC as well.}

  \begin{restatable}{lemma}{tplanarsccon}\label{l:planar-sccon}
    Let $G$ be a planar digraph subject to planarity-preserving edge insertions and deletions.
    One can maintain whether $G$ has a single SCC in $\Ot(n^{2/3})$ worst-case time per update.
  \end{restatable}
Similarly, one could ask how \emph{dynamic \#SSR} (i.e., counting vertices reachable from a single source) relates to dynamic SCCs in planar digraphs.
Especially since:
\begin{enumerate}[label=(\arabic*)]
  \item in general directed graphs, dynamic SCCs and dynamic \#SSR  currently have matching lower-~\cite{AbboudW14, BrandNS19} and upper bounds~\cite{KarczmarzS23, Sankowski05} (up to polylogarithmic factors);
  \item the former problem is at least as hard as the latter in the sense that dynamic \#SSR reduces to dynamic SCCs in general graphs easily\footnote{Consider the graph $G'$ obtained from $G$ by adding a supersink $t$ with a single outgoing edge $ts$ and incoming edges $vt$ for all $v\in V$. Then $v\in V$ is reachable
from $s$ in $G$ iff $s$ and $v$ are strongly connected in $G'$. See also~\cite{BrandNS19}.}, whereas
an opposite reduction is not known.
\end{enumerate}
Unfortunately, the aforementioned reduction of dynamic \#SSR to dynamic SCCs breaks planarity rather badly.
Interestingly, the \emph{path net} technique we develop to obtain Theorem~\ref{t:planar-sccs} does not
seem to work for counting ``asymmetric'' reachabilities from a single source.

Nevertheless, we observe that the Voronoi diagram machinery developed for computing the diameter of a planar graph~\cite{GawrychowskiKMS21}
almost immediately yields a more efficient data structure for dynamic \#SSR in planar digraphs with
$\Ot(n^{4/5})$ update time; see Lemma~\ref{t:planar-ssr-count} in Section~\ref{s:ssr}.
It is worth noting that Voronoi diagrams-based techniques (as developed for distance oracles~\cite{GawrychowskiMWW18}) have been used in the
pairwise strong connectivity data structure~\cite{Charalampopoulos22}. However, as we discuss later
on, it is not clear how to apply those for the dynamic SCCs problem. This is why
Theorem~\ref{t:planar-sccs} relies on a completely different path net approach developed in this paper.

\paragraph{Organization.}
We review some standard planar graph tools in Section~\ref{s:prelims}.
Then, as a warm-up, we show the data structure for dynamic SC in Section~\ref{s:planar-conn}.
In Section~\ref{s:planar-sccs} we define the path net data structure and show how it can be
used to obtain a dynamic SCCs data structure. In Section~\ref{s:ds-planar} we describe
the path net data structure in detail. Finally, Section~\ref{s:ssr} is devoted to the fully dynamic \#SSR problem in planar digraphs.

\section{Preliminaries}\label{s:prelims}
  In this paper we deal with \emph{directed} graphs.
  We write $V(G)$ and $E(G)$ to denote the sets of vertices and edges of $G$, respectively. We omit $G$ when the graph in consideration is clear from the context.
  A graph $H$ is a \emph{subgraph} of $G$, which~we~denote by $H\subseteq G$, iff $V(H)\subseteq V(G)$ and $E(H)\subseteq E(G)$.
  We write $e=uv\in E(G)$ when referring to edges of $G$.
  By $\rev{G}$ we denote $G$ with edges reversed.

  A sequence of vertices $P=v_1\ldots v_k$, where $k\geq 1$, is called
  an $s\to t$ path in~$G$ if $s=v_1$, $v_k=t$ and there is an edge $v_iv_{i+1}$ in $G$ for each $i=1,\ldots,k-1$.
  We sometimes view a path $P$ as a subgraph of $G$ with vertices $\{v_1,\ldots,v_k\}$
  and (possibly zero) edges $\{v_1v_2,\ldots,v_{k-1}v_k\}$.
  For convenience, we sometimes consider a single edge $uv$ a path.
  If $P_1$ is a $u \to v$ path and $P_2$ is a $v \to w$ path, we denote by $P_1\cdot P_2$ (or simply $P_1P_2$) a path obtained by concatenating $P_1$ with $P_2$.
  A vertex $t\in V(G)$ is \emph{reachable} from $s\in V(G)$ if there is an $s\to t$ path in $G$.
  \newcommand{\dstr}{\mathcal{D}}

\paragraph{Planar graph toolbox.}
An \emph{$r$-division}~\cite{DBLP:journals/siamcomp/Frederickson87} $\rdiv$ of a planar graph, for $r \in [1,n]$,
is a decomposition of the graph into a union of $O(n/r)$ pieces $P$, each of size $O(r)$ and with $O(\sqrt{r})$ boundary vertices
(denoted~$\bnd{P}$), i.e., vertices shared with some other
piece of $\rdiv$.
We denote by $\bnd{\rdiv}$ the set $\bigcup_{P\in\rdiv}\bnd{P}$.
If additionally $G$ is plane-embedded, all pieces are connected, and the boundary vertices of each piece~$P$ of the $r$-division~$\rdiv$ are distributed
among $O(1)$ faces of $P$ that contain the vertices from $\bnd{P}$ \emph{exclusively} (also called holes of $P$), we call $\rdiv$ an \emph{$r$-division with few holes}.
\cite{KleinMS13} showed that an $r$-division with few holes of a triangulated graph can be computed in linear time.

\paragraph{Fully dynamic $r$-divisions.}
Many dynamic algorithms for planar graphs maintain
$r$-divisions and useful piecewise auxiliary data structures under dynamic updates.
Let us slightly generalize the definition of an $r$-division with few holes
to non-planar graphs by dropping the requirement that $G$ as a whole is planar but
retaining all the other requirements (in particular, the individual pieces
are plane-embedded).

\begin{restatable}{theorem}{tdynrdiv}\label{t:dyn-rdiv}{\normalfont\cite{Charalampopoulos22, KleinS98, Subramanian93}}
  Let $G=(V,E)$ be a planar graph that undergoes edge deletions and edge insertions (assumed to preserve the planarity of $G$). 
  Let $r\in [1,n]$.
  
  There is a data structure maintaining an $r$-division with few holes $\rdiv$ of some $G^+$, where $G^+$ can be obtained from $G$ by adding edges\footnote{Note that $G^+$ need not be planar.}, such that each piece $P \in \rdiv$
  is accompanied with some auxiliary data structures that
      can be constructed in $T(r)$ time given~$P$ and use $S(r)$ space.
  
  The data structure uses $O\left(n+\frac{n}{r}\cdot S(r)\right)$ space and can be initialized
  in $O\left(n+\frac{n}{r}\cdot T(r)\right)$ time.
  After each edge insertion/deletion, it
  can be updated in $O(r+T(r))$ \emph{worst-case} time.
\end{restatable}

\section{Fully dynamic SC data structure}\label{s:planar-conn}

To illustrate the general approach and introduce some of the concepts used for obtaining Theorem~\ref{t:planar-sccs},
in this section we first prove Lemma~\ref{l:planar-sccon}. That is, we show that the information whether
a planar graph $G$ is strongly connected can be maintained in $\Ot(n^{2/3})$ time per update.

We build upon the following general template used previously for designing fully dynamic
data structures supporting reachability, strong connectivity, and shortest paths queries in planar graphs, e.g.,~\cite{Subramanian93, KleinS98, FR, Charalampopoulos22, KaplanMNS17}.
As a base, we will maintain dynamically an $r$-division with few holes $\rdiv$ of $G$ using Theorem~\ref{t:dyn-rdiv}
with auxiliary piecewise data structures to be fixed later.
Intuitively, as long as the piecewise data structures are powerful enough
to allow recomputing the requested graph property (e.g., strong connectivity, shortest
path between a fixed source/target pair) while spending $r^{1-\eps}$ time
per piece, for some choice of $r$ we get a sublinear update bound
of $\Ot(n/r^{\eps}+r+T(r))$. For example, if $T(r)=O(r^{9.9})$ and $\eps=0.1$, for $r=n^{0.1}$ we get $\Ot(n^{0.99})$
worst-case update time bound.

\paragraph{Reachability certificates.}
\cite{Subramanian93} described \emph{reachability certificates} that sparsify reachability
between a subset of vertices lying on $O(1)$ faces of a plane digraph $G$ into a (non-necessarily planar)
digraph of size near-linear in the size of the subset in question. Formally, we have the following.
\begin{lemma}[\cite{Subramanian93}]\label{l:subramanian}
  Let $H$ be a plane digraph with a distinguished set $\bnd{H}\subseteq V(H)$ lying on some $O(1)$ faces of $H$.  There exists a directed graph $X_H$, where $\bnd{H}\subseteq V(X_H)$,
  of size $\Ot(|\bnd{H}|)$ satisfying the following property:
  for any $u,v\in \bnd{H}$, a path $u\to v$ exists in $H$
  if and only if there exists a $u\to v$ path in $X_H$.
  The graph $X_H$ can be computed in $\Ot(|H|)$ time.
\end{lemma}
\begin{remark}\label{rem:sub}
  For Lemma~\ref{l:subramanian} to hold, it is enough that $\bnd{H}$ lies on $O(1)$ Jordan curves in the plane, each of them having the embedding of $H$ entirely (but not strictly) on one side of the curve. In particular, it is enough that $\bnd{H}$ lies on $O(1)$ faces of some plane supergraph $H'$ with $H\subseteq H'$.
\end{remark}

Roughly speaking, \cite{Subramanian93}~uses reachability certificates as auxiliary data structures
in Theorem~\ref{t:dyn-rdiv} in order to obtain
a fully dynamic reachability data structure.
Crucially, the union of the piecewise certificates preserves pairwise reachability
between the boundary vertices $\bnd{\rdiv}$, or more formally
(see e.g.~\cite{Charalampopoulos22} for a proof):
\begin{restatable}{lemma}{bndcert}\label{l:bnd-cert}
For any $u,v\in \bnd{\rdiv}$,
  $u$ can reach $v$ in $G$ if and only if $u$ can reach $v$ in $X=\bigcup_{P\in\rdiv} X_P$.
\end{restatable}

\paragraph{Strong connectivity data structure.}

The union of certificates $X$ preserves reachability, and thus strong connectivity between the
vertices $\bnd{\rdiv}:=\bigcup_{P\in\rdiv}\bnd{P}$. As a result, if $G$ is strongly connected,
then so is $\bnd{\rdiv}$ in~$X$. But the reverse implication might not hold.
It turns out that for connected graphs, to have an equivalence, it is enough to additionally maintain, for each piece $P$,
whether $P$ is strongly connected conditioned on whether $\bnd{P}$ is strongly connected in $G$.

In the following, we give a formal description of the data structure.
As already said, the data structure maintains a dynamic $r$-division $\rdiv^+$ of a supergraph $G^+$ of $G$ (i.e., the input graph), as given by Theorem~\ref{t:dyn-rdiv}. Since $G\subseteq G^+$, the pieces $\{P^+\cap G:P^+\in \rdiv^+\}$
induce an $r$-division $\rdiv$ of~$G$; however, a piece $P\in\rdiv$ is not necessarily connected and $\bnd{P}$ not necessarily lie on $O(1)$ faces of $P$, so $\rdiv$ is not technically an $r$-division with few holes. Nevertheless, $\bnd{P}$ still lies on $O(1)$ faces of a connected plane supergraph $P^+$ of $P$ that do not contain
vertices outside $\bnd{P}$. Consequently, by Remark~\ref{rem:sub}, we can still use Lemma~\ref{l:subramanian} to construct a sparse reachability certificate for the piece $P\in\rdiv$. For obtaining Lemma~\ref{l:planar-sccon}, we do not
require anything besides that, so 
for simplicity and wlog. we can assume we work with $\rdiv$ instead of $\rdiv^+$.

While $\rdiv$ evolves, each piece $P$ is accompanied with a reachability certificate $X_P$ of Lemma~\ref{l:subramanian}. Note that since $|\bnd{P}|=O(\sqrt{r})$, $X_P$ has size $\Ot(\sqrt{r})$ and can be constructed
in $\Ot(r)$ time.
Moreover, for each $P$, let $C_{\bnd{P}}$ be a directed simple cycle on the vertices $\bnd{P}$. We additionally store the (1-bit) information whether the graph $P\cup C_{\bnd{P}}$ is strongly connected.
Clearly, this can be computed in $O(|P|)=O(r)$ time.
All the accompanying data structures of a piece $P\in \rdiv$ can be thus constructed in $\Ot(r)$ time.
Therefore, by Theorem~\ref{t:dyn-rdiv}, they are maintained in $\Ot(r)$ time per update.

Finally, in a separate data structure, we maintain whether $G$ is connected
(in the undirected sense). This can be maintained within $n^{o(1)}$ worst-case update
time deterministically even in general graphs~\cite{GoranciRST21};
in our case, also a less involved data structure such as~\cite{Frederickson85} would suffice.

After $\rdiv$ and the accompanying data structures are updated, strong connectivity of $G$ can be verified as follows.
First of all, the union $X$ of all $X_P$, $P\in \rdiv$, is formed.
Note that we can test whether the vertices $\bnd{\rdiv}$ are strongly connected
in $X$ in $O(|X|)=\Ot(n/\sqrt{r})$ time
by computing the strongly connected components $\sccs_X$ of $X$ using any classical linear time algorithm.
If $G$ is not connected, or $\bnd{\rdiv}$ is not strongly connected in~$X$, we declare $G$ not strongly connected.
If, on the other hand, $\bnd{\rdiv}$ is strongly connected in $X$, we simply check whether
$P\cup C_{\bnd{P}}$ is strongly connected for each $P\in \rdiv$ and if so, declare $G$ strongly connected. This takes $O(n/r)$ time.
Thus, testing strong connectivity takes $\Ot(n/\sqrt{r})$ time.
The following lemma establishes the correctness.

\begin{restatable}{lemma}{lsubramanian}\label{l:subramanian-correct}
  $G$ is strongly connected if and only if $G$ is connected, the vertices $\bnd{\rdiv}$ are strongly connected in $X$, and for all $P\in\rdiv$, $P\cup C_{\bnd{P}}$ is strongly
  connected.
\end{restatable}
\begin{proof}
  First suppose that $G$ is strongly connected. Then, $G$ is clearly connected.
  Moreover, by Lemma~\ref{l:bnd-cert}, $\bnd{\rdiv}$ is strongly connected
  in $X$.
  For contradiction, suppose that for some $P\in \rdiv$ and $u,v\in V(P)$, $u$ cannot reach $v$ in
  $P\cup C_{\bnd{P}}$. By strong connectivity of $G$, there exists some path $Q=u\to v$ in $G$.
  Since $u$ cannot reach $v$ in $P$, $Q$ is not fully contained in $P$. As a result, $Q$
  can be expressed as $Q_1\cdot R\cdot Q_2$, where $Q_1=u\to a,Q_2=b\to v$ are fully contained in $P$,
  and $a,b\in \bnd{P}$. But there is a path $Z=a\to b$ in $C_{\bnd{P}}$, so there is
  a $u\to v$ path $Q_1\cdot Z\cdot Q_2$ in $P\cup C_{\bnd{P}}$, a contradiction.

  Now consider the ``$\impliedby$'' direction. 
  Suppose $G$ is connected, the vertices $\bnd{\rdiv}$ are strongly connected in $X$, and for all $P\in\rdiv$, $P\cup C_{\bnd{P}}$ is strongly
  connected. By Lemma~\ref{l:bnd-cert}, $\bnd{\rdiv}$ is strongly connected in $G$.
  Consider any $P\in \rdiv$ and let $x,y\in V(P)$.
  We first prove that there exists a path $x\to y$ in $G$.
  Indeed, if an $x\to y$ path exists in $P$, it also exists in $G$.
  Otherwise, since $P\cup C_{\bnd{P}}$ is strongly connected, there exists a
  path $Q=x\to y$ in $P\cup C_{\bnd{P}}$ that can be expressed
  as $Q_1\cdot R\cdot Q_2$, where $Q_1=x\to a$ and $Q_2=b\to y$ are 
  fully contained in $P$ and $a,b\in \bnd{P}$.
  But since $a,b\in \bnd{\rdiv}$, by strong connectivity of $\bnd{\rdiv}$, there exists a path $R'=a\to b$ in $G$.
  Since $Q_1,Q_2\subseteq G$, $Q_1\cdot R'\cdot Q_2$ is an $x\to y$ path in~$G$.
  
  Now take arbitrary $u,v\in V(G)$. 
  If there exists a piece in $\rdiv$ containing both $u$ and $v$, then we have
  already proved that there exists a path $u\to v$ in $G$.
  Otherwise, let $P_u,P_v$, $P_u\neq P_v$, be some pieces of $\rdiv$ containing $u,v$, respectively.
  We have $P_u\neq G$ and $P_v\neq G$. Since $G$ is connected, $P_u$
  has at least one boundary vertex $a\in \bnd{P_u}$.
  Similarly, $P_v$ has at least one boundary vertex $b\in \bnd{P_v}$.
  We have proved that there exist paths $u\to a$ and $b\to v$ in $G$.
  But also $a,b\in \bnd{\rdiv}$, so by the strong connectivity of $\bnd{\rdiv}$, there exists a path $a\to b$ in $G$ as well.
  We conclude that there exists a path $u\to v$ in $G$.
  Since $u,v$ were arbitrary, $G$ is indeed strongly connected.
\end{proof}

The worst-case update time of the data structure is $\Ot(r+n/\sqrt{r})+n^{o(1)}$. By setting $r=n^{2/3}$,
we obtain the following.
\tplanarsccon*

\section{Dynamic strongly connected components}\label{s:planar-sccs}
The approach we take for maintaining strong connectivity in planar graphs does not easily generalize
even to dynamic SCCs counting.
This is the case for the following reason. Even if the piece $P$ is fixed (static),
there can be possibly an exponential number of different
assignments of the vertices $\bnd{P}$ to the SCCs in $G$ (when the other pieces are subject
to changes), whereas for dynamic SC, a non-trivial situation arises only when all of $\bnd{P}$
lies within a single SCC.
In order to achieve sublinear update time,
for any assignment we need to be able to count the SCCs
fully contained in~$P$ in time sublinear in $r$
after preprocessing $P$ in only \emph{polynomial} (and not exponential) time.

The following notion will be crucial for all our developments.
\begin{definition}
  Let $P\in\rdiv$, and let $A\subseteq \bnd{P}$. A \emph{path net} $\Pi_P(A)$ induced by $A$ is the set of vertices of~$P$ that lie on some directed path in $P$ connecting some two elements of $A$.
\end{definition}
In other words, the path net $\Pi_P(A)$ contains vertices $v\in V(P)$ such that $v$ can reach $A$ and can be reached from $A$ in $P$. We call a path net $\Pi_P(A)$ \emph{closed} if $A=\Pi_P(A)\cap \bnd{P}$, that is, there are no boundary vertices of $P$ outside $A$ that can reach and can be reached from $A$.

The following key lemma relates a piece's path net to the SCCs of $G$.

\begin{lemma}\label{l:scc-reach}
  Let $S$ be an SCC of $G$ containing at least one boundary vertex of $P$,
  i.e., $S\cap \bnd{P}\neq\emptyset$.
  Then the path net $\Pi_P(S\cap \bnd{P})$ is closed and equals $S\cap V(P)$.
\end{lemma}
\begin{proof}
  Let us first argue that $\Pi_P(S\cap \bnd{P})$ is closed.
  If it was not, there would exist $b\in \bnd{P}\setminus S$ such that there exist
  paths $b\to (S\cap \bnd{P})$ and $(S\cap \bnd{P})\to b$ in $P$.
  It follows that $b$ can reach and be reached from $S$ in $G$, i.e., $b$
  is strongly connected with $S$. Hence, $b\in S$, a contradiction.

  Let $v\in \Pi_P(S\cap \bnd{P})$.
  Since $v$ can reach and can be reached from $S\cap\bnd{P}$ in $P$,
  then it is indeed strongly connected with $S$ in $G$, since the vertices $S$ are strongly connected in~$G$.
  So $v\in S\cap V(P)$.

  Now let $v\in S\cap V(P)$.
  Pick any $b\in S\cap\bnd{P}$ (possibly $b=v$ if $v\in \bnd{P}$). There exists
  paths $R=v\to b$ and $Q=b\to v$ in~$G$. Note that
  $R$ has some prefix $R_1=v\to a$ that is fully contained in~$P$ and $a\in \bnd{P}$.
  Similarly, $Q$ has a suffix $Q_1=c\to v$ that is fully contained in $P$ and $c\in\bnd{P}$.
  Since there exists paths $v\to a$, $a\to b$, $b\to c$, $c\to v$ in $G$, vertices $a,b,c$
  are strongly connected in~$G$. So $a,c\in S\cap \bnd{P}$. The paths $R_1,Q_1$
  certify that $v$ can be reached from and can reach $S\cap\bnd{P}$ in $P$.
  Therefore, $v\in \Pi_P(S\cap \bnd{P})$ as desired.
\end{proof}

If an SCC $S$ is as in Lemma~\ref{l:scc-reach}, then since the vertices $\bnd{P}$ might be shared with other pieces of~$\rdiv$, $\Pi_P(S\cap \bnd{P})\setminus \bnd{P}$
constitutes the vertices of $S$ contained \emph{exclusively in the piece~$P$}.
As there are only $\Ot(n/\sqrt{r})$ boundary vertices through all pieces, their affiliation
to the SCCs of $G$ can be derived from the (SCCs of the) certificate graph $X=\bigcup_{P\in\rdiv} X_P$ (defined and maintained as in Section~\ref{s:planar-conn}),
i.e., they may be handled efficiently separately.
Consequently, being able to efficiently aggregate labels -- or report the elements -- of the sets of the form
$\Pi_P(A)\setminus \bnd{P}$ (where $\Pi_P(A)$ is closed) is the key to obtaining an efficient implicit representation
of the SCCs of~$G$. 
Our main technical contribution (Theorem~\ref{t:ds-planar}) is the \emph{path net data structure} enabling precisely that. The data structure requires a rather large $\Ot(r^3)$ preprocessing time
but achieves the goal by supporting queries about $A\subseteq \bnd{P}$ in near-optimal $\Ot(|A|)$ time. Formally, we show:

\begin{restatable}{theorem}{tdsplanar}\label{t:ds-planar}
  Let $P\in \rdiv$ and let $\alpha:V(P)\to\mathbb{R}$ be a weight function.
  In $\Ot(r^{3})$ time one can construct a data structure satisfying the following.

  Given a subset $A\subseteq \bnd{P}$ such that $\Pi_P(A)$ is closed,  in $\Ot(|A|)$ time one can:
  \begin{itemize}[topsep=4pt,itemsep=2pt]
    \item create an \emph{iterator} that enables listing elements of $\Pi_P(A)\setminus \bnd{P}$ in $O(1)$ time per element,
    \item aggregate weights over $\Pi_P(A)\setminus \bnd{P}$, i.e., compute $\sum_{v\in \Pi_P(A)\setminus \bnd{P}}\alpha(v)$.
  \end{itemize}
\end{restatable}
\begin{remark}
  We do not require using subtractions to compute the aggregate weights.
  In fact, the data structure of Theorem~\ref{t:ds-planar} can be easily modified to aggregate weights coming from any
  semigroup, e.g., one can compute the max/min weight in $\Pi_P(A)\setminus \bnd{P}$ within these bounds.
\end{remark}

Our high-level strategy is to maintain the certificates and path net data structures accompanying individual pieces along with the $r$-division.
Roughly speaking, to obtain the information about the SCCs of $G$ beyond how the partition of $\bnd{\rdiv}$ into SCCs looks like, we will query the path net data structures for each piece $P$ with the sets $A$ equal to the SCCs of $X$ having non-empty intersection with $\bnd{P}$.
We prove Theorem~\ref{t:ds-planar} later on, in Section~\ref{s:ds-planar}.

In the remaining part of the section, we explain in detail
how, equipped with Theorem~\ref{t:ds-planar}, a dynamic (implicit) strongly connected components data structure can be obtained.
As in Section~\ref{s:planar-conn}, we maintain an $r$-division $\rdiv$ dynamically,
and maintain sparse reachability certificates $X_P$
along with the set $\sccs_X$ of SCCs of $X=\bigcup_{P\in\rdiv} X_P$.
Moreover, for each $P$ we store the strongly connected components $\sccs_P$ of $P$.
Let $\sccs_{\bnd{P}}$ be the elements of $\sccs_P$ that contain a boundary vertex,
and $\sccs_{P\setminus\bnd{P}}$ the elements of $\sccs_P$ that do not.
Clearly, we have $\sccs_P=\sccs_{\bnd{P}}\cup \sccs_{P\setminus\bnd{P}}$ and
$\left(\bigcup \sccs_P\right)\cap \left(\bigcup\sccs_{P\setminus\bnd{P}}\right)=\emptyset$.

For each piece $P\in \rdiv$, we additionally store a path net data structure $\mathcal{D}_P$ of Theorem~\ref{t:ds-planar}
with an appropriately defined weight function (to be chosen depending on the application later).
Note that for a piece $P$, all the auxiliary data structures accompanying $P$ that we have defined can be computed in $\Ot(r^{3})$ time.
We now consider the specific goals that can be achieved this way.

\paragraph{Finding the largest SCC.} Denote by $S^*$ the largest SCC of $G$. To be able to identify
$S^*$, and e.g., compute its size, we proceed as follows.
For each piece $P$, we additionally maintain the largest
SCC $S^*_P$ of $P$. The sizes of all the SCCs of $P$, in particular the size of $S^*_P$,
can be easily found and stored after computing $\sccs_P$.

Note that if the largest SCC $S^*$ of $G$ is not contained entirely in any individual piece $P$
(and thus is larger than $\max_{P\in\rdiv}|S^*_P|$),
it has to intersect $\bnd{\rdiv}$.
More specifically, in this case for each piece $P$ such that $S^*\cap V(P)\neq\emptyset$,
we have $S^*\cap \bnd{P}\neq\emptyset$.

Recall that by Lemma~\ref{l:bnd-cert}, $X=\bigcup_{P\in\rdiv} X_P$ preserves the strong connectivity relation
on the vertices $\bnd{\rdiv}$.
Therefore, if $S^*$ intersects $\bnd{\rdiv}$, it has to contain $B\cap \bnd{\rdiv}$ for some SCC $B$ of $X$.
For any such $B$, we can compute the size of the SCC $S_B$ of $G$ satisfying $B\cap\bnd{\rdiv}\subseteq S_B$
as follows. First of all, $|S_B\cap \bnd{\rdiv}|=|B\cap \bnd{\rdiv}|$ since $B$ is an SCC of $X$.
It is thus enough to compute, for all $P\in\rdiv$, $|S_B\cap (V(P)\setminus \bnd{P})|$.
Since the sets $V(P)\setminus \bnd{P}$ are pairwise disjoint across the pieces, by adding these values, we will get the desired size $|S_B|$.

We have already argued that if $B\cap \bnd{P}=\emptyset$, then $S_B\cap V(P)=\emptyset$.
If, on the other hand, $B\cap\bnd{P}$ is non-empty, by Lemma~\ref{l:scc-reach}, if we use the weight function $\alpha(v)\equiv 1$ in the piecewise
data structures $\mathcal{D}_P$ of Theorem~\ref{t:ds-planar}, we can compute
  $|S_B\cap (V(P)\setminus \bnd{P})|=\sum_{v\in \Pi_P(B\cap \bnd{P})\setminus\bnd{P}}\alpha(v)$
in $\Ot(|B\cap \bnd{P}|)$ time using the input set $A:=B\cap \bnd{P}$.
We conclude that the sizes $S_B$ for all $B\in\sccs_X$ can be computed in
time
\begin{equation}\label{eq:sum-over-pieces}
  \Ot\left(\sum_{B\in\sccs_X}\sum_{\substack{P\in\rdiv\\B\cap\bnd{P}\neq\emptyset}}|B\cap\bnd{P}|\right)=\Ot\left(\sum_{P\in\rdiv}\sum_{\substack{B\in\sccs_X\\B\cap\bnd{P}\neq\emptyset}}|B\cap\bnd{P}|\right)=\Ot\left(\sum_{P\in\rdiv} |\bnd{P}|\right)=\Ot(n/\sqrt{r}).
\end{equation}
Finally, $S^*$ is either equal to the largest $S_B$ for $B\in \sccs_X$, or
the largest $S^*_P$ (through $P\in\rdiv$).
The latter is the case if $\max_{B\in \sccs_X}|S_B|<\max_{P\in\rdiv}|S^*_P|$.
Which case we fall into is easily decided once all the $O(n/\sqrt{r})$ sizes $|S_B|$ are computed.

\paragraph{Accessing the SCC of a specified vertex.}
Suppose first that we know the SCC $S_v$ of $G$ containing a query vertex $v$,
and additionally whether $S_v$ intersects $\bnd{\rdiv}$.
If $S_v\cap \bnd{\rdiv}=\emptyset$, then $v$ is a vertex of a unique piece $P$,
and $S_v\in \sccs_{P\setminus \bnd{P}}$.
In this case, we can clearly compute the size of $S_v$ and report the elements
of $S_v$ in $O(1)$ time since~$S_v\in\sccs_P$ is stored explicitly.

Otherwise, if $S_v\cap \bnd{\rdiv}\neq\emptyset$, we reuse the information that
we have computed for finding the largest SCC. 
We have already described how to compute the sizes of all the SCCs $S$ of $G$
intersecting $\bnd{\rdiv}$ (in particular $S_v$)
along with these respective intersections
in $\Ot(n/\sqrt{r})$ time upon update.
Moreover, for all $P\in\rdiv$ we have computed
$|S\cap (V(P)\setminus \bnd{P})|$ using the data structure $\mathcal{D}_P$ of Theorem~\ref{t:ds-planar}. As a result, we can store,
for each such $S$, a subset $L(S)\subseteq \rdiv$ of pieces $P$ such that
$S\cap (V(P)\setminus \bnd{P})\neq\emptyset$.
Recall that for all $P\in L(S)$, we can also use $\mathcal{D}_P$ to create
an iterator for reporting the elements of $S\cap (V(P)\setminus \bnd{P})$ in $O(1)$ time per element.
So, in order to efficiently report elements of any such $S$,
we first report the vertices of $S\cap \bnd{\rdiv}$, and then the elements of each
$S\cap (V(P)\setminus \bnd{P})$ for subsequent pieces $P\in L(S)$.
Indeed, one needs only $O(1)$ worst-case time to find each subsequent vertex of $S$.

Finally, we are left with the task of of finding the SCC $S_v$ of a query vertex $v\in V$.
It is not clear how to leverage path net data structures for that
in the case when $S_v$ intersects $\bnd{\rdiv}$.
Instead, we use the data structure
of~\cite{Charalampopoulos22} in a black-box way.
That data structure handles fully dynamic edge updates in $\Ot(n^{4/5})$ worst-case time,
and provides $O(\log^2{n})$ worst-case time access to consistent (for queries issued between any two subsequent updates) SCC identifiers
of individual vertices\footnote{The query time of that data structure can be easily reduced to $O(\log{n}\cdot\log\log{n})$
without affecting the $\Ot(n^{4/5})$ update bound
if one simply replaces the classical MSSP data structure~\cite{MSSP}
used internally for performing point location queries in additively weighted Voronoi diagrams~\cite{GawrychowskiMWW18}
with the MSSP data structure of~\cite{LongP21}.}.
Using~\cite{Charalampopoulos22}, after every update we find the identifiers $I$ of
the SCCs of vertices $\bnd{\rdiv}$ in $G$ in $\Ot(n/\sqrt{r})$ time.
Now, to find the SCC of $v$ upon query, we find
the identifier $i_v$ of the SCC containing~$v$ in $O(\log^2{n})$ time.
If $i_v\in I$ (which can be tested in $O(\log{n})$ time), we obtain that $v$ is in an SCC of~$G$ intersecting
$\bnd{\rdiv}$ and some vertex $b_v$ from the intersection. $b_v$ can be in turn used to access
$L(S_v)$ and thus enable reporting the elements of $S_v$.
Otherwise, if $i_v\notin I$, $S\cap\bnd{\rdiv}=\emptyset$ and thus $S$ equals
the unique SCC from $S_{P\setminus\bnd{P}}$ containing $v$ in the unique piece $P$ containing $v$.

\paragraph{Counting strongly connected components.} Let us separately count SCCs $\sccs_{\bnd{\rdiv}}$ that intersect
$\bnd{\rdiv}$ and those that do not.
The former can be counted in $\Ot(n/\sqrt{r})$ time by counting
the SCCs of $X$ that intersect $\bnd{\rdiv}$ (that we maintain).
The latter can be computed as follows. Consider the sum $\Phi=\sum_{P\in \rdiv}|\sccs_{P\setminus\bnd{P}}|$.
If we wanted the sum $\Phi$ to count the SCCs not intersecting $\bnd{\rdiv}$,
then an SCC $S\in \sccs_{P\setminus\bnd{P}}$ contributes to the sum
unnecessarily precisely when $S$ is not an SCC of $G$.
To see that, note that if an SCC $S$ of $P$ is not an SCC of $G$, it has
to be a part of another SCC $S'$ of $G$ that also contains vertices
of other pieces, i.e., $S'$ intersects $\bnd{\rdiv}$.
From Lemma~\ref{l:scc-reach}, we conclude:

\begin{corollary}\label{cor:scc-reach}
  An SCC $S\in \sccs_{P\setminus\bnd{P}}$ is not an SCC of $G$ iff there exists
  (precisely one) SCC $B$ of~$X$ such that $S\subseteq \Pi_P(B\cap \bnd{P})$ (or equivalently, such that $S\cap \Pi_P(B\cap\bnd{P})\neq\emptyset$).
\end{corollary}
As a result, we can count the number of SCCs in $G$ that do not
intersect $\bnd{\rdiv}$ by subtracting from $\Phi$, for each $P\in\rdiv$,
and each $B\in \sccs_X$
the number $c_{P,B}$ of SCCs in $\sccs_{P\setminus\bnd{P}}$ that
intersect $\Pi_P(B\cap \bnd{P})$.
To this end, we can use the data structure $\mathcal{D}_P$ of Theorem~\ref{t:ds-planar} built
upon~$P$ (and maintained as described before) with a weight function $\alpha$ on $V(P)$
assigning $1$
to an arbitrary single vertex $v_S$ of each SCC $S\in \sccs_{P\setminus\bnd{P}}$,
and $0$ to all other vertices.
By Corollary~\ref{cor:scc-reach}, with such a weight function, $c_{P,B}=\sum_{v\in \Pi_P(B\cap\bnd{P})\setminus\bnd{P}}\alpha(v)$ can be computed in $\Ot(|B\cap \bnd{P}|)$ time using $\mathcal{D}_P$.
Consequently, similarly as in~\eqref{eq:sum-over-pieces}, over all $B\in \sccs_X$, and $P\in\rdiv$, computing all the values $c_{P,B}$ will take $\Ot(n/\sqrt{r})$
time.
As mentioned before, the SCC count is obtained by subtracting those from $\Phi$
and adding the result to the count of $\sccs_{\bnd{\rdiv}}$.

Depending on the application, the worst-case update time of the data structure is $\Ot(n/\sqrt{r}+r^{3})$ or
$\Ot(n/\sqrt{r}+r^{3}+n^{4/5})$. The bound is optimized for $r=n^{2/7}$ and this yields the following.
\planarsccs*

\section{The path net data structure}\label{s:ds-planar}
This section is devoted to describing the below key component
of our dynamic SCCs data structure.
\tdsplanar*

\subsection{Overview}
\cite{Charalampopoulos22}~showed that for any SCC $S$ of $G$,
$V(P)\cap S$ forms an intersection of two cells coming from two
carefully prepared additively weighted Voronoi diagrams on $P$ with sites $\bnd{P}$.
As a result, they could use Voronoi diagram point location mechanism~\cite{GawrychowskiMWW18, LongP21}
for testing in $O(\polylog{n})$ time whether a query vertex lies in such an intersection.
If we tried to follow this approach, we would need to be able to aggregate/report vertices in such
intersections of cells coming from two seemingly unrelated Voronoi diagrams.
This is very different from just testing membership, and it is not clear whether this can be done efficiently.

Instead, in order to prove Theorem~\ref{t:ds-planar}, we take a more direct approach. 
As is done typically, we first consider the situation when $A$ lies on a single face of $P$.
In the single-hole case, the first step is to reduce to the case when the input piece $P$ is acyclic;
note that if a vertex lies in the path net $\Pi_P(A)$, its entire SCC in $P$ does.
Acyclicity and appropriate perturbation~\cite{Erickson18} allows us to
pick a collection of paths $\pi_{u,v}$, for all $u,v\in \bnd{P}$, such that every two paths
in the collection are either disjoint or their intersection forms a single subpath of both.
This property makes the paths~$\pi_{u,v}$ particularly convenient to use for cutting
the piece $P$ into smaller non-overlapping parts.

More specifically, the paths $\pi_{u,v}$ are used to partition -- using a polygon triangulation-like procedure -- a queried net $\Pi_P(A)$
into regions in the plane with vertices $B\subseteq A$ bounded by either fragments of the face of $P$ 
containing $\bnd{P}$ or some paths $\pi_{u,v}$ for $u,v\in B$ (so-called \emph{base instances}).
A~base instance has a very special structure guaranteeing that for a given vertex $v\in V(P)\setminus \bnd{P}$,
there are \emph{only $O(1)$ pairs} $s,t\in B$ such that an $s\to v\to t$ path exists in $P$.
At the end, this crucial property can be used to reduce a base instance query $B$
even further to looking up $\Ot(|B|)$ preprocessed answers for base instances with at most $5$ vertices from $\bnd{P}$,
of which there are at most $\Ot(|\bnd{P}|^5)=\Ot(r^{5/2})$.
The precomputation of small base instances can be done in $\Ot(r^3)$ time.

For efficiently implementing the polygon triangulation-like partition procedure -- which
repeatedly cuts off base instances from the ``core'' part of the problem --
we develop a dynamic data structure for existential reachability
queries on a single face of a plane digraph (Lemma~\ref{l:any-reach}, proved in Section~\ref{a:anyreachproof}).
To this end, we leverage the single-face reachability characterization of~\cite{ItalianoKLS17}.

As the final step, we show a reduction from the case when $A$ can lie on $k$ faces on $P$
to the case when $A$ lies on $k-1$ faces. The reduction -- which eventually reduces the problem
to the single-hole case -- can blow up the piece's size by a constant factor. However, since there are only $O(1)$ holes
initially, the general case can still be solved only a constant factor slower than the single-hole case.

\subsection{Reducing to the acyclic case}\label{s:acyclic}
Recall that $P$ is a piece of an $r$-division with few holes $\rdiv$ and thus has size $O(r)$. Moreover, $\bnd{P}$ lies on $O(1)$
faces (\emph{holes}) of some supergraph $P^+$ of $P$ and $|\bnd{P}|=O(\sqrt{r})$.
For simplicity, we will not differentiate between the faces/holes of $P$ and $P^+$. Whenever
we refer to a hole $h$ of $P$, we mean a closed curve $h$ in the plane such that
whole embedding of $P$ lies weakly on a single side of~$h$, and only vertices $\bnd{P}$ may lie
on the curve $h$.

First of all, we compute the strongly connected components $\sccs_P$ of $P$.
Observe that from the point of view of supported queries, all vertices
of a single SCC $S\in\sccs_P$ are treated exactly the same:
when a single vertex $v\in S\setminus \bnd{P}$ is reported (or its weight $\alpha(v)$ is aggregated),
all other vertices from $S\setminus \bnd{P}$ are as well.
For each $S\in\sccs_P$ we precompute the aggregate weight $\alpha(S)=\sum_{v\in S\setminus\bnd{P}}\alpha(v)$
in $O(r)$ time.
We then contract each strongly connected component $S\in\sccs_P$
into a single vertex~$v_S$. Crucially, since each subgraph $P[S]$ is connected in the undirected
sense, contractions can be performed in an embedding preserving way (via a sequence of single-edge contractions), so that:
\begin{itemize}
  \item The obtained graph $P'$ is acyclic.
  \item The vertices $\bnd{P'}:=\{v_S:S\in\sccs_P, S\cap \bnd{P}\neq\emptyset\}$ lie on $O(1)$
    faces of (a supergraph of) $P'$.
\end{itemize}
Indeed, to see the latter property, label each $b\in\bnd{P}$ initially each with hole that $b$ lies on
and make non-boundary vertices unlabeled.
When the vertices merge, label the resulting vertex using one of the involved vertices' labels, if any of them has one. At the end there will be still $O(1)$ labels,
each vertex of $\bnd{P'}$ will have one of these labels, and for each label, all the vertices
holding that label will be incident to a single face of $P'$.
Clearly, $|P'|=O(r)$ and $|\bnd{P'}|=O(\sqrt{r})$.

\begin{restatable}{lemma}{lacyclic}\label{l:acyclic}
Suppose a data structure $\mathcal{D}'$ of Theorem~\ref{t:ds-planar} is built for the acyclic graph $P'$
with the weight function $\alpha'(v_S):=\alpha(S)=\sum_{v\in S\setminus \bnd{P}}\alpha(v)$.
Then, a query $A\subseteq \bnd{P}$ from Theorem~\ref{t:ds-planar} can be reduced, in $O(|A|)$ time,
  to a query about $A'=\{v_S:S\in\sccs_P,S\cap A\neq\emptyset\}$ issued to $\mathcal{D}'$.
\end{restatable}
\begin{proof}
By the assumption that $\Pi_P(A)$ is closed, we have that $\Pi_P(A)$
contains only such $S\in \sccs_P$ that $S\cap \bnd{P}\subseteq A$.
If the set $\Pi_{P'}(A')$ contained some $v_{R}\in \bnd{P'}\setminus A'$, then
in the graph $P$ a vertex from $R\cap (\bnd{P}\setminus A)$ could reach and could
be reached from $A$, which would contradict $\Pi_P(A)\cap \bnd{P}=A$.
As a result, $A'$ is indeed a valid query parameter to the data structure $\mathcal{D}'$
and in $\Ot(|A'|) \le \Ot(|A|)$ time we can gain reporting access to $\Pi_{P'}(A')\setminus \bnd{P'}$
or aggregate the weights of $\Pi_{P'}(A')\setminus \bnd{P'}$.
Observe that any vertex $v_S\in \Pi_{P'}(A')\setminus \bnd{P'}$
corresponds to a subset $S\subseteq \Pi_P(A)\setminus \bnd{P}$.
These subsets $S$, being strongly connected components of $P$, are clearly pairwise disjoint.
So the reporting procedure, when $\mathcal{D'}$ reports $v_S$, can simply
report the individual vertices of $S$ one by one.
Moreover, the weight $\alpha(v_S)$ equals the aggregate weight $\sum_{v\in S}\alpha(v)$.
It is easy to see that this way, using $\mathcal{D'}$ we can actually efficiently report/aggregate all the vertices
of $(\Pi_P(A)\setminus \bnd{P})\cap \left(\bigcup \{S:v_S\in \Pi_{P'}(A')\setminus \bnd{P'}\}\right)$,
or, in other words, the vertices of $\Pi_P(A)\setminus \bnd{P}$ that are not strongly
connected to a boundary vertex of $P$.
However, some vertices from 
$\Pi_P(A)\setminus \bnd{P}$ are not necessarily in 
$\bigcup \{S:v_S\in \Pi_{P'}(A')\setminus \bnd{P'}\}$.
But these are precisely the non-boundary vertices of $P$ that are strongly connected
to some vertex of $A$ in $P$.
As a result, we can handle these remaining vertices by iterating over the $O(|A|)$ SCCs of $P$ that intersect $A$.
\end{proof}

In the following we will assume that $P$ is an acyclic graph.
We will no longer need the assumption $\Pi_P(A)\cap \bnd{P}=A$; this was only
needed for the efficient reduction to the acyclic case.
Consequently, we will design a data structure with query time $\Ot(|A|)$
even if $A\subsetneq \Pi_P(A)\cap \bnd{P}$.

In the remaining part of this section, we assume we only want to handle queries where the query set $A$ lies on a
\emph{single hole}~$h$ of an acyclic piece $P$. We present the reduction
of the general case to the single-hole case in Section~\ref{s:multiple-hole}.
Our strategy for the single-hole case will be to efficiently break the problem into
subproblems for which we have the answer precomputed.

\subsection{Picking non-crossing paths.}
We first fix, for any $u,v\in V(P)$ such that $u$ can reach $v$,
one particular $u\to v$ path $\pi_{u,v}\subseteq P$.
We apply the perturbation scheme of~\cite{Erickson18} to $P$,
so that the shortest paths in $P$ become unique.
For any $u,v\in V(P)$ such that $u$ can reach $v$ in $P$,
we define $\pi_{u,v}$ to be the unique shortest $u\to v$ path in~$P$.
Note that each $\pi_{u,v}$ is a simple path in $P$.
We have the following crucial property.

\begin{lemma}\label{l:crossing}
  Let $u,v,x,y\in V(P)$ be such that $u$ can reach $v$ and $x$ can reach $y$ in $P$.
  If \linebreak
  $V(\pi_{u,v})\cap V(\pi_{x,y})\neq\emptyset$, then $\pi_{u,v}$ and $\pi_{x,y}$
  share a single (possibly zero-edge) subpath.
\end{lemma}
\begin{proof}
  Let $a$ ($b$) be the first (last, resp.) vertex on $\pi_{u,v}$ that is also a vertex of $\pi_{x,y}$.
  $\pi_{u,v}$ can be expressed as $Q_1\cdot T\cdot Q_2$, where $T=a\to b$.
  Moreover, $V(\pi_{u,v})\cap V(\pi_{x,y})\subseteq V(T)$.
  If $a\neq b$, the vertex~$b$ cannot appear before $a$ on $\pi_{x,y}$, because then $a$ and $b$ would lie on
  a cycle in $P$, contradicting acyclicity of $P$.
  As a result, $\pi_{x,y}$ has an $a\to b$ subpath as well. 
  Since shortest paths in $P$ are unique and shortest paths have optimal substructure, the $a\to b$ subpath 
  of $\pi_{x,y}$ equals $T$.
  So, $\pi_{x,y}$ can be expressed as $R_1\cdot T\cdot R_2$.
  Since $\pi_{x,y}$ is simple, and $V(\pi_{u,v})\cap V(\pi_{x,y})\subseteq V(T)$,
  we indeed have $\pi_{u,v}\cap \pi_{x,y}=T$.
\end{proof}

Observe that since $P$ is a DAG, for $u,v\in V(P)$, $u\neq v$, $\pi_{u,v}$ and $\pi_{v,u}$ cannot
both exist. When one of them exists, we will sometimes write $\pi_{\{u,v\}}$ to denote
the one that exists.

\subsection{Generalization}
Before we proceed, we need to state our problem in a more general way that will enable
decomposition into smaller subproblems of similar kind.
Let us first fix a counterclockwise order of $V(h)$, as follows.
Let $u_1,e_1,u_2,e_2,\ldots,u_r,e_r$ be the sequence of vertices and edges on a counterclockwise bounding walk of $h$ with an arbitrary starting point. The encountered vertices $u_i$ might not be all distinct since~$h$ can be a non-simple face in general. The order $\prec_h$ of $V(h)$ that we use is obtained by removing duplicate vertices from the sequence in an arbitrary way, i.e., for each $v\in V(h)$ we pick one index $i_v$ such that $u_{i_v}=v$. If $x,y\in V(h)$ and $i_x<i_y$, then we define $C_{x,y}$ to be the curve $e_{i_x}e_{i_x+1}\ldots e_{i_{y}-1}$. Otherwise,
if $i_x>i_y$, then $C_{x,y}=e_{i_x}e_{i_x+1}\ldots e_re_1e_2\ldots e_{i_{y}-1}$.
If $h$ is simple, then $C_{x,y}$ is a simple curve.
However, in general, $C_{x,y}$ can be self-touching but it does not cross itself.

Suppose $B\subseteq \bnd{P}\cap V(h)$ is given ordered by $\prec_h$.
Moreover, for each pair of neighboring (i.e., appearing consecutively in the \emph{cyclic} order of $B$ given by $\prec_h$) vertices $x,y$ of~$B$,
a curve $\Phi_{x,y}$ connecting $x$ to $y$ is given.
The curve $\Phi_{x,y}$ equals either $\pi_{\{x,y\}}$ or $C_{x,y}$.
Let $\Phi$ be the closed curve formed by concatenating the subsequent
curves $\Phi_{x,y}$ for all neighboring $x,y\in B$.
For brevity, we also extend the notation $\Phi_{a,b}$ to non-neighboring 
vertices $a,b$ of $B$, and define $\Phi_{a,b}$ to be a concatenation
of the curves $\Phi_{x,y}$ for all neighboring pairs $(x,y)$ between $a$ and $b$ on $\Phi$.
We don't require $\Phi$ to be a simple (Jordan) curve; parts of it may be overlapping.
Whereas~$\Phi$ is allowed to be self-touching, it does not cross itself,
like e.g., the cycle bounding a non-simple face of a plane graph.
Let $P[\Phi]$ denote the region of $P$ that lies \emph{weakly} inside the curve $\Phi$.

The objective of the problem $(B,\Phi)$ is to aggregate weights of (or report) the vertices of ${\Pi_P(B)\setminus \bnd{P}}$ that
lie \emph{strictly inside} $\Phi$. 
More formally, we want to aggregate vertices $v$ that lie, at the same time, strictly on the left side
of each of the curves $\Phi_{x,y}$ (seen as a curve directed from $x$ to $y$) 
where $x,y$ are neighboring elements of $B$ sorted by $\prec_h$.\footnote{Even more formally, assuming $P$ is embedded in such a way that $h$ is the infinite face of a supergraph of $P$, we are interested in the vertices lying in the intersection
of the regions strictly inside closed curves $\Phi_{x,y}\cdot C_{y,x}$ for all neighboring $x,y$ in $B$.}
Note that with such a defined problem, the original goal of the query procedure
can be rephrased as $(A,h)$ since all the vertices of $\Pi_P(A)\setminus\bnd{P}$ indeed lie
strictly inside~$h$.

\subsection{Preprocessing}

\paragraph{Preprocessing for small instances and subproblems.}
\begin{itemize}
  \item For any tuple $B=(b_1,\ldots,b_q)$ of at most $4$ vertices of $\bnd{P}$ appearing
in that order in $\prec_h$, and any out of at most $2^q=O(1)$ possible curves $\Phi$
that might constitute the problem $(B,\Phi)$, we precompute the aggregate weight
    and the list of vertices in $\Pi_P(B)\setminus \bnd{P}$ that lie strictly inside $\Phi$.
This can clearly be done in $\Ot(|\bnd{P}|^4\cdot |P|)=\Ot(r^3)$ time.

  \item For any 5-tuple $\tau=(b_0,\ldots,b_4)$ ordered by $\prec_h$, and any out of at most $2^4$ possible
    curves $\Phi_{b_i,b_{i+1}}\in \{\pi_{\{b_i,b_{i+1}\}},C_{b_i,b_{i+1}}\}$, where $i\in \{0,\ldots,3\}$,
    we precompute and store the set $X_\tau$ of vertices $x\in V(P)\setminus \bnd{P}$
    such that:
    \begin{enumerate}
      \item $x$ lies on some path connecting $b_1$ and $b_2$ in $P$,
      \item $x$ \emph{does not} lie on any path connecting $b_0$ and $b_1$ in $P$,
      \item $x$ lies strictly to the left of each $\Phi_{b_i,b_{i+1}}$ for $i=\{0,\ldots,3\}$.
    \end{enumerate}
    We also store the aggregate weight $w(X_\tau)$.
    This preprocessing can be performed in a brute-force way in $\Ot(|\bnd{P}|^5\cdot |P|)=\Ot(r^{7/2})$ time.
    It can be also optimized fairly easily to $\Ot(r^3)$ time and space, as follows.

    First, for every $a\in \bnd{P}$, and $v\in V(P)\setminus \bnd{P}$, find the earliest vertex $z_{a,v}\in \bnd{P}$
    following $a$ in the order $\prec_h$ such that $v$ does not lie
    strictly to the left of $\pi_{\{a,z_{a,v}\}}$.
    Note that having fixed $a\in \bnd{P}$, the property of $v$ being strictly left of $\pi_{\{a,b\}}$ is monotonous
    in the vertex $b\in \bnd{P}$ following~$a$.
    The vertices $z_{a,v}$ can be computed in $\Ot(|\bnd{P}|^2\cdot |P|^2)=\Ot(r^3)$ naively.

    Having
    fixed $(b_0,\ldots,b_3)$, compute the set $X'$ of all vertices $x\in V(P)\setminus \bnd{P}$
    satisfying all the above conditions possibly except the
    one that $x$ lies strictly to the left of $\Phi_{b_3,b_4}$.
    Computing all such $X'$ takes $\Ot(|\bnd{P}|^4\cdot |P|)=\Ot(r^3)$ time.
    Now, for each of the $\Ot(r^2)$ obtained sets~$X'$, in $\Ot(r)$ time we sort the elements $x\in X'$
    by the value $z_{b_3,x}$ where $z_{b_3,x}$ ($b_3\prec_h z_{b_3,x}$) appearing later in $\prec_h$ is considered smaller.
    This way, for each possible $b_4\in B$, the desired set of vertices~$X_\tau$ forms a prefix of the vertices $X'$:
    we have $x\in X_\tau$ iff $x\in X'$ and $z_{b_3,x}$ appears later than $b_4$ (and after $b_3$) in the order $\prec_h$.
    Storing $X'$ sorted like that, along with the aggregate prefix weights, allows us to aggregate/report
    any set $X_\tau$ for a 5-tuple $\tau=(b_0,\ldots,b_4)$ and 4-tuple $(\Phi_{b_i,b_{i+1}})_{i=0}^3$
    in $O(1)$ time after only $\Ot(r^3)$ preprocessing.

  \item For any $u,v\in \bnd{P}$, we store the path $\pi_{\{u,v\}}$ itself, along
    with aggregate weights of all its subpaths. This data can be computed in $\Ot(|\bnd{P}|^2\cdot |P|^2)=\Ot(r^3)$
    time in a brute-force way.

  \item Finally, for any two pairs $(u,v),(x,y)\in \bnd{P}\times \bnd{P}$, we compute the intersection
    of the paths $\pi_{\{u,v\}}$ and $\pi_{\{x,y\}}$. By Lemma~\ref{l:crossing}, that intersection
    is either empty or is a subpath of both these paths.
    Hence, it is enough to store the two endpoints of the intersection subpath only.
    The desired intersections can be found in $\Ot(|\bnd{P}|^4\cdot |P|)=\Ot(r^3)$ time
    in a brute-force way.

 \end{itemize}

\newcommand{\anyds}{\mathcal{L}}

\paragraph{Existential reachability data structure.}
We also build data structures $\anyds,\anyds^R$ of the following lemma for $P$ and $P^R$, respectively.
\begin{restatable}{lemma}{lanyreach}\label{l:any-reach}
  In $\Ot(r)$ time one can construct a data structure
  maintaining an (initially empty) set $Z\subseteq \bnd{P}\cap V(h)$ and supporting the following operations in $O(\polylog{n})$ time:
  \begin{itemize}
    \item insert or delete a single $b\in \bnd{P}\cap V(h)$ either to or from $Z$.
    \item for any query vertex $v\in (\bnd{P}\cap V(h))\setminus Z$, find any $z\in Z$ (if exists) that $v$ can reach in $P$.
  \end{itemize}
\end{restatable}
\noindent The proof of Lemma~\ref{l:any-reach} is deferred to Section~\ref{a:anyreachproof}.

\subsection{Answering queries}
We now proceed with the description of our algorithm solving the general problem
$(B,\Phi)$.
\paragraph{Base case.}
Let the elements of $B=\{b_1,\ldots,b_k\}$ be sorted according to their order $\prec_h$ on~$\Phi$.
For convenience, identify $b_{k+l}$ with $b_l$ for any integer $l$.

We first consider the easier \emph{base case}, with the following additional requirement:
\begin{quote}
  \textit{For any $1\leq i<j\leq k$, if $b_i$ can reach $b_j$ or can be reached from $b_j$ in $P$, then either $j=i+1$ or $(i,j)=(1,k)$.}
\end{quote}

If $k\leq 4$, we will simply return the precomputed aggregate weight, or an iterator
to a list of vertices in $\Pi_P(B)\setminus\bnd{P}$ strictly inside $\Phi$.
So suppose $k\geq 5$. We start with the following lemma.
\begin{restatable}{lemma}{lauxbase}\label{l:aux-base}
  Let $v\in \Pi_P(B)\setminus \bnd{P}$ lie strictly inside $\Phi$. Then:
  \begin{itemize}
    \item There exists such $i$ that
      $v$ lies on some $b_i\to b_{i+1}$ or on some $b_{i+1}\to b_i$ path in $P$.
    \item Moreover, there is at most one more pair $\{x,y\}\subseteq B$,
      $\{x,y\}\neq\{b_i,b_{i+1}\}$
      such that $v$ lies on some $x\to y$ or $y\to x$ path in $P$,
      and either $\{x,y\}=\{b_{i-1},b_i\}$, or $\{x,y\}=\{b_{i+1},b_{i+2}\}$.
  \end{itemize}
\end{restatable}
\begin{proof}
Item~(1) follows easily by the additional requirement of the base case.
Let $v$ lie on some $b_j\to c$ path in $P$, where $c\in B$. Since
$P$ is acyclic, we have that $b_j\neq c$ and thus $c\in \{b_{j-1},b_{j+1}\}$.
  If $c=b_{j+1}$, we put $i=j$ and if $c=b_{j-1}$, we put $i=j-1$.

  For item~(2), consider the case when $v$ lies on some $b_i\to b_{i+1}$; the case when $v$ lies on a $b_{i+1}\to b_i$ path is symmetric. Suppose $v$ also lies on some $x\to y$ path in $P$, where $x,y\in B$, $x\neq y$, $\{x,y\}\neq \{b_i,b_{i+1}\}$.
Since $b_i$ can reach $y$ through $v$, $y\in \{b_{i-1},b_{i+1}\}$ by the requirement.
  Similarly, since $x$ can reach $b_{i+1}$ through $v$, $x\in \{b_i,b_{i+2}\}$.
  We cannot have $x=b_{i+2}$ and $y=b_{i-1}$ since $b_{i+2}$ reaching $b_{i-1}$ is a contradiction with the first base case
  requirement for $k\geq 5$.
  As a result, $(x,y)\neq (b_i,b_{i+1})$ implies $(x,y)=(b_i,b_{i-1})$ or $(x,y)=(b_{i+2},b_{i+1})$.
Moreover, $v$ cannot lie on both some $b_i\to b_{i-1}$ path
and some $b_{i+2}\to b_{i+1}$ path, since then $b_{i+2}$ could reach $b_{i-1}$,
which is again a contradiction for $k\geq 5$.
\end{proof}

By Lemma~\ref{l:aux-base}, each $v\in \Pi_P(B)\setminus \bnd{P}$
falls into exactly
one of the $k$ sets $Y_i$, for $i=1,\ldots,k$,
such that $Y_i$ contains those $v$ in $V(P)\setminus \bnd{P}$
that lie on a path in $P$ connecting $b_{i}$ and $b_{i+1}$ (in any direction),
but at the same time \emph{do not} lie on any path connecting $b_{i-1}$ and $b_i$ in $P$ (in any direction).
Indeed, if~$v$ appears only on paths connecting $b_i$ and $b_{i+1}$,
it will be included in $Y_i$.
On the other hand, if $v$ appears both on paths connecting $b_i$ and $b_{i+1}$
and on paths connecting $b_{i+1}$ and $b_{i+2}$,
it will be included only in $Y_i$, but not in $Y_{i+1}$.

\begin{restatable}{lemma}{lfivepaths}\label{l:fivepaths}
  Suppose $v\in Y_i$. Then $v$ lies strictly
  inside $\Phi$ if and only if $v$ lies strictly on the left side of $\Phi_{b_{i-1},b_i}$, $\Phi_{b_i,b_{i+1}}$, $\Phi_{b_{i+1},b_{i+2}}$,
  and $\Phi_{b_{i+2},b_{i+3}}$.
\end{restatable}
\begin{proof}
  Since strictly inside $\Phi$ means strictly to the left of all $\Phi_{b_j,b_{j+1}}$, the ``$\implies$'' direction is trivial.

  Consider the ``$\impliedby$'' direction.
  For contradiction, suppose $v$ does not lie strictly inside~$\Phi$.
  Then, for some $j\notin \{i-1,i,i+1,i+2\}$, $v$ lies weakly
  to the right of $\Phi_{b_j,b_{j+1}}$.
  Since $v\in V(P)\setminus \bnd{P}$ and the hole $h$ contains only vertices
  of $\bnd{P}$, this means that $\Phi_{b_j,b_{j+1}}=\pi_{\{b_j,b_{j+1}\}}$.
  Since $b_j$ and $b_{j+1}$ are consecutive in $B$, and $i\notin \{j,j+1\}$,
  $b_i$ lies weakly to the left of $\pi_{\{b_j,b_{j+1}\}}$.
  
  By $v\in Y_i$, the vertex $v$ lies on a path $Q$ connecting $b_i$ and $b_{i+1}$ in $P$.
  Since $v$ and $b_i$ lie weakly on different sides of
  $\pi_{\{b_j,b_{j+1}\}}$, the path $Q$ has to cross
  $\pi_{\{b_j,b_{j+1}\}}$ at a vertex $w$ appearing not later than $v$ on $Q$ (possibly $v=w$).
  If $Q=b_i\to b_{i+1}$, the existence of $w$ implies that 
  there exists an $s\to b_{i+1}$ path in $P$ going through $v$
  such that $s\in \{b_j,b_{j+1}\}$. By $v\in Y_i$, this implies $s\in \{b_i,b_{i+2}\}$.
  As a result, $j\in \{i-1,i,i+1,i+2\}$, a contradiction.
  Similarly, if $Q=b_{i+1}\to b_i$, there exists an $s\to b_i$ path in $P$
  going through $v$ such that $s\in \{b_j,b_{j+1}\}$.
  But on the other hand by $v\in Y_i$ we have $s=b_{i+1}$ so $j\in \{i,i+1\}$, a contradiction.
\end{proof}

As a result, we can equivalently aggregate vertices $v$ in each $Y_i$ under a (seemingly) weaker
requirement that $v$ lies strictly on the left side of 
$\Phi_{b_{i-1},b_{i}}$, $\Phi_{b_i,b_{i+1}}$, $\Phi_{b_{i+1},b_{i+2}}$, and $\Phi_{b_{i+2},b_{i+3}}$.
But this is, again, part of the precomputed data for the tuple $(b_{i-1},b_{i},b_{i+1},b_{i+2},b_{i+3})$.

Consequently, there are $k$ disjoint sets $Y_i$ to consider. We can thus aggregate weights or
construct a list for reporting vertices from $v\in \Pi_P(B)\setminus \bnd{P}$ strictly
inside $\Phi$ in $O(k)$ time.
We thus obtain:
\begin{lemma}\label{l:base-time}
  The base case can be handled in $O(|B|)$ time.
\end{lemma}

\paragraph{General case.}
To solve the general case, we reduce it to a number of base case instances.
To this end, we maintain a partition of $B=S\cup T$ such that $S$ precedes $T$ on $\Phi$.
Let $S=\{s_1,\ldots,s_p\}$ and $T=\{t_1,\ldots,t_q\}$.
We will gradually simplify the problem while maintaining the following invariants:
\begin{enumerate}[label=(\arabic*)]
  \item For any $u,v\in B$, if $u$ can reach $v$ in $P$, then $\pi_{u,v}\subseteq P[\Phi]$.
  \item For any $1\leq i<j\leq p$, if $s_i$ can reach $s_j$ or can be reached from $s_j$, then $j=i+1$.
  \item If $x,y\in B$ are neighbors in the counterclockwise order $\prec_h$ on $\Phi$ and $\Phi_{x,y}=\pi_{\{x,y\}}$,
    then $x,y\in S$.
  \item In the data structures $\anyds,\anyds^R$ of Lemma~\ref{l:any-reach}, the set $Z$ satisfies $Z=S\setminus\{s_p\}$.
\end{enumerate}

The algorithm will gradually modify $B,S,T,\Phi$ until we have $S=B$ and $T=\emptyset$.
Note that when $T=\emptyset$, $(B,\Phi)=(S,\Phi)$ satisfies the requirement of
the base case and can be solved in $\Ot(|B|)$ time.

Initially we put $\Phi=h$, $S=\{b_1,b_2\}$, and insert $b_1$ to $Z$
to satisfy the invariants.

The main loop of the procedure runs while $T\neq\emptyset$ and does the following. Using $\anyds$ and $\anyds^R$,
in $O(\polylog{n})$ time
we test whether $t_1$ can reach $Z$ or can be reached from $Z$ in $P$.
If this is not the case, we simply move $t_1$ to the end of $S$, an update
$Z$ accordingly. Note that $|T|$ decreases then.

Otherwise, we can find all vertices $X=\{x_1,\ldots,x_l\}\subseteq Z$ that $t_1$ can reach or can be reached
from in $\Ot(|X|)$ by repeatedly extracting them from the data structures $\anyds,\anyds^R$.
Additionally we sort $X$ so that the order $x_1,\ldots,x_l$ matches the order of $S$ on $\Phi$.
Let $\pi_i$ denote the path $\pi_{\{t_1,x_i\}}$ possibly reversed to go from $t_1$ to $x_i$,
and $\pi_i^R$ denote the reverse $\pi_i$ going from $x_i$ to~$t_1$.
The vertices $X$ are used to ``cut off'' $l$ base case instances, as follows.
For $i=0,\ldots,l$, let $S_i$ denote the vertices of $S$ between $x_i$ and $x_{i+1}$ on $\Phi$ (inclusive),
where we set $x_{l+1}:=s_p$, $x_0:=s_1$.
We split the problem $(B,\Phi)$ into the following subproblems (see Figure~\ref{f:split} for better understanding):
\begin{enumerate}[topsep=2pt]
  \item For each $i=1,\ldots,l-1$, an instance $(S_i\cup \{t_1\}, \Phi_{x_i,x_{i+1}}\cdot \pi_{i+1}^R\cdot \pi_i)$.
  \item An instance $(S_l\cup \{t_1\}, \Phi_{x_l,t_1}\cdot \pi_l)$.
  \item An instance $(S_0\cup T, \Phi_{s_1,x_1}\cdot \pi_1^R \cdot \Phi_{t_1,s_1})$.
\end{enumerate}
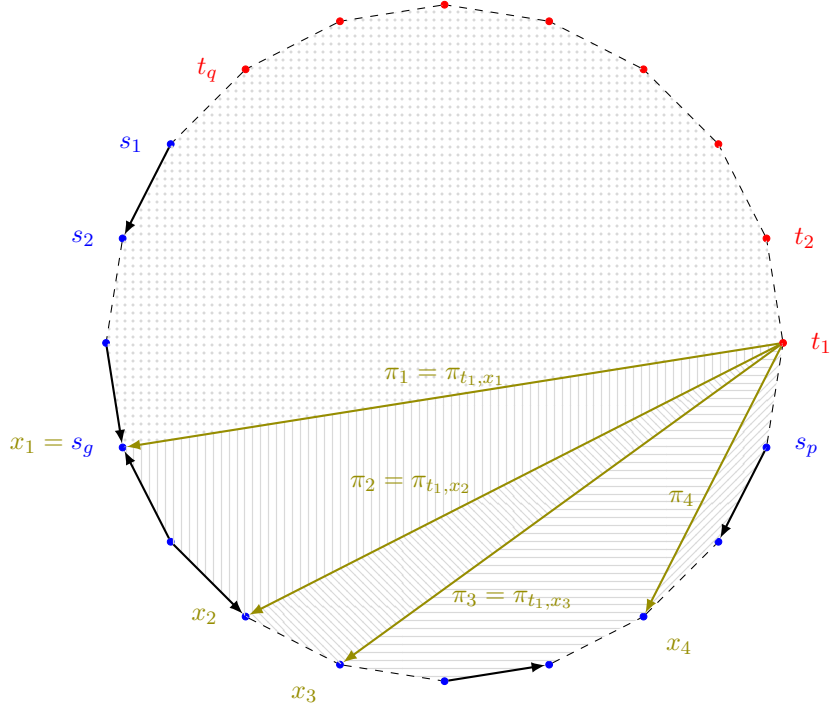
\begin{figure}[ht!]
  \centering

\begin{tikzpicture}[scale=0.9, every node/.style={scale=0.9}]
\foreach \a in {1,2,...,12}{
  \node[draw,circle,fill,inner sep=1pt,color=blue] (s\a) at (\a*360/20+126: 5) {};
}
\foreach \a in {1,2,...,8}{
  \node[draw,circle,fill,inner sep=1pt,color=red] (t\a) at (\a*360/20-18: 5) {};
}
  \node[left=0.2 of s1,blue] {$s_1$};
  \node[left=0.2 of s2,blue] {$s_2$};
  \node[left=0.2 of s4,olive] {$x_1={\color{blue}s_g}$};
  \node[left=0.2 of s6,olive] {$x_2$};
  \node[below left=0.2 of s7,olive] {$x_3$};
  \node[below right=0.2 of s10,olive] {$x_4$};
  \node[right=0.2 of s12,blue] {$s_p$};
  \node[right=0.2 of t1,red] {$t_1$};
  \node[right=0.2 of t2,red] {$t_2$};
  \node[left=0.2 of t8,red] {$t_q$};

  \draw[pattern=vertical lines,pattern color=black!15,draw=none] (s4.center) to (s5.center) to (s6.center) to (t1.center) to (s4.center);

  \draw[pattern=north west lines,pattern color=black!15,draw=none] (s6.center) to (s7.center) to (t1.center) to (s6.center);
  \draw[pattern=horizontal lines,pattern color=black!15,draw=none] (s7.center) to (s8.center) to (s9.center) to (s10.center) to (t1.center) to (s7.center);

  \draw[pattern=north east lines,pattern color=black!15,draw=none] (s10.center) to (s11.center) to (s12.center) to (t1.center) to (s10.center);
  
  \draw[pattern=dots,pattern color=black!15,draw=none] (t1.center) to (t2.center) to (t3.center) to (t4.center) to (t5.center) to (t6.center) to (t7.center) to (t8.center) to (s1.center) to (s2.center) to (s3.center) to (s4.center) to (t1.center);

  \draw[-latex,thick,olive] (t1) to (s4);
  \draw[-latex,thick,olive] (t1) to (s6);
  \draw[-latex,thick,olive] (t1) to (s7);
  \draw[-latex,thick,olive] (t1) to (s10);

  \node[olive] at (0,-0.5) {$\pi_1=\pi_{t_1,x_1}$};
  \node[olive] at (-0.5,-2.1) {$\pi_2=\pi_{t_1,x_2}$};
  \node[olive] at (1,-3.8) {$\pi_3=\pi_{t_1,x_3}$};
  \node[olive] at (3.5,-2.3) {$\pi_4$};

  \draw[-latex,thick] (s1) to (s2);
  \draw[-latex,thick] (s3) to (s4);
  \draw[-latex,thick] (s5) to (s4);
  \draw[-latex,thick] (s5) to (s6);
  \draw[-latex,thick] (s8) to (s9);
  \draw[-latex,thick] (s12) to (s11);

  \draw[dashed] (s2) to (s3);
  \draw[dashed] (s6) to (s7);
  \draw[dashed] (s7) to (s8);
  \draw[dashed] (s9) to (s10);
  \draw[dashed] (s10) to (s11);
  \draw[dashed] (s12) to (t1);

  \draw[dashed] (t1) to (t2);
  \draw[dashed] (t2) to (t3);
  \draw[dashed] (t3) to (t4);
  \draw[dashed] (t4) to (t5);
  \draw[dashed] (t5) to (t6);
  \draw[dashed] (t6) to (t7);
  \draw[dashed] (t7) to (t8);
  \draw[dashed] (t8) to (s1);

\end{tikzpicture}
  \caption{Splitting the instance $(B,\Phi)$, where $B=S\cup T$,  into 5 smaller instances with paths $\pi_1,\ldots,\pi_l$ (either originating or ending in $t_1$) for $l=4$. The vertices $S=\{s_1,\ldots,s_p\}$ are shown in blue, whereas the vertices $T=\{t_1,\ldots,t_q\}$ in red. The black
  arrows and dashed lines represent the individual parts of the curve $\Phi$: paths of the form $\pi_{u,v}$ or parts of the curve $h$, respectively. Note that the black arrows appear only on the $\Phi_{s_1,s_p}$ part of $\Phi$. The vertices
  $x_1,\ldots,x_4$, marked green, are precisely all the vertices of $S\setminus\{s_p\}$ that $t_1$ can reach or can be reached from. The obtained smaller instances are marked with distinct patterns. The instances marked with line patterns (types 1 or 2) are base instances. The instance marked using a dotted pattern (type 3) might constitute the only obtained instance that is not a base instance (for which the algorithm continues).} \label{f:split}
\end{figure}
\begin{lemma}\label{l:req1}
  For each of the above subproblems $(B',\Phi')$, and $u,v\in B'$, if $u$ can reach $v$ in $P$,
  then $\pi_{u,v}\subseteq P[\Phi']$.
\end{lemma}
\begin{proof}
  Note that $(B',\Phi')$ is obtained from $(B,\Phi)$ by cutting it out of $(B,\Phi)$
  with at most two paths $\pi_{\{t_1,a\}},\pi_{\{t_1,b\}}$, for some
  $t_1,a,b\in B'\subseteq B$.
  By our assumption, $\pi_{u,v}\subseteq P[\Phi]$.
  As a result, if $\pi_{u,v}$ was not contained in $P[\Phi']$,
  then $\pi_{u,v}$ would need to cross either $\pi_{\{t_1,a\}}$ or $\pi_{\{t_1,b\}}$.
  However, this is impossible by Lemma~\ref{l:crossing} and since $u,v\in V(P[\Phi'])$.
\end{proof}
\begin{lemma}
  The obtained instances of types 1 and 2 above fall into the base case.
\end{lemma}
\begin{proof}
  To see that the base case requirement is satisfied, recall that by the invariant posed on $S$,
  if two elements of $S_i$, where $i\in \{1,\ldots,l\}$,
  are related (wrt. reachability in $P$), they are neighboring in $S_i$.
  By construction, $t_1$ can be only related to the first and the last element of $S_i$.
\end{proof}
Since the cutting is performed using non-crossing paths in $P$, the regions $P[\Phi']$ for the obtained subproblems can only share
their boundaries, that is, if some $v$ is strictly inside one of the curves~$\Phi'$,
then it it is not strictly inside another obtained curve~$\Phi''$.
Therefore, if we proceeded with the above subproblems recursively, we would
not aggregate or report any vertex of $v\in \Pi_P(B)\setminus \bnd{P}$ twice.
However, we still need to report/aggregate vertices that lie on paths
$\pi_1,\ldots,\pi_l$ strictly inside the curve $\Phi$.
We now discuss how this strategy is implemented.
Let us first consider solving the subproblems recursively.
We handle all the obtained base case instance of types 1 and 2
as explained before. If $x_1=s_g$, then by Lemma~\ref{l:base-time}, the total time required for this
is $O\left(\sum_{i=1}^l(|S_i|+1)\right)=O((p - g)+l)$.
But note that $l\leq p-g$, so in fact the bound is $O(p-g)$.
       
To handle the instance $(S_0\cup T, \Phi_{s_1,x_1}\cdot \pi_1^R \cdot \Phi_{t_1,s_1})$, we simply
replace $(B,\Phi)$ with it and proceed with solving it using the algorithm for the general case.
To this end, we
set $S:=S_0\cup \{t_1\}$, $T:=\{t_2,\ldots,t_q\}$ and update $Z$ in the data structures $\anyds,\anyds^R$
to $S_0$ by removing elements.
Then, invariant (1) is satisfied by Lemma~\ref{l:req1}, and invariants (2), (3) and (4) are satisfied by construction.
This way, in $\Ot(\Delta)$ time we reduce the instance $(B,\Phi)$ by $\Delta=|S|-|S_0|=p-g$
vertices. Recall that $Z=S\setminus\{s_p\}$ implies that $g<p$. Thus, $\Delta\geq 1$ and
the sizes of $B$ and $T$ strictly decrease.

Let us now discuss how to aggregate/report the vertices of 
$\Pi_P(B)\setminus \bnd{P}$ that lie on any of the paths $\pi_1,\ldots,\pi_l$
(that are not handled in any of the subproblems), but at the same
time lie strictly inside $\Phi$ (before altering $(B,\Phi)$).
Since $\Phi$ is formed of either the edges of the hole $h$, or from the paths $\pi_{\{u,v\}}$, and each $\pi_i$
is contained in $P[\Phi]$, equivalently we need to aggregate the vertices
of $\Pi_P(B)\setminus \bnd{P}$ on the paths $\pi_1,\ldots,\pi_l$ that do
not lie on $\Phi$.

Observe that since $x_1,\ldots,x_l$ lie on $\Phi$ in that order,
and the paths $\pi_1,\ldots,\pi_l$ all lie in $P[\Phi]$ and are non-crossing by Lemma~\ref{l:crossing},
for any three $i<j<k$, we have $V(\pi_i)\cap V(\pi_k)\subseteq V(\pi_j)$.
As a result, any vertex on these paths
is included in precisely one of the sets $V(\pi_i)\setminus V(\pi_{i-1})$,
for $i=1,\ldots,l$ and $V(\pi_0):=\emptyset$.
Moreover, in Lemma~\ref{l:atmost5} we will show that each $\pi_i$ can possibly have a non-empty intersection
with $O(1)$ parts (between neighboring elements of $B$) of $\Phi$,
that we can also identify in $O(1)$ time.
Since, by Lemma~\ref{l:crossing}, for every path $\pi_{\{u,v\}}$,
the intersection of $\pi_{\{u,v\}}$ with $\pi_i$ is either empty
or forms a subpath of $\pi_i$, 
aggregating or reporting the required vertices of $\pi_i$
boils down to aggregating or reporting the vertices
of $V(P)\setminus \bnd{P}$ on some $O(1)$ subpaths of $\pi_i$ that form
what remains from $\pi_i$ after removing $O(1)$ of its intersections with other paths $\pi_{u,v}$.

\begin{restatable}{lemma}{latmostfive}\label{l:atmost5}
  Consider the moment when the split into subproblems happens.
  Suppose $x_i=s_j$. Let $u,v\in B$ be neighboring on $\Phi$,
  so that $u$ comes before $v$ in the counterclockwise order $\prec_h$ on $\Phi$.
  Then $(V(\pi_i)\cap V(\Phi_{u,v}))\setminus\bnd{P}\neq\emptyset$ implies
  that $u=s_{j'}$ for some $j'\in \{1,\ldots,p\}$ such that $|j-j'|\leq 2$.
\end{restatable}
\begin{proof}
  Recall that the curve is $\pi_i$ is either the path $\pi_{\{x_i,t_1\}}$ or its reverse.
  Let us only consider the case when $\pi_{\{x_i,t_1\}}=\pi_{x_i,t_1}$; the
  case when $\pi_{\{x_i,t_1\}}=\pi_{t_1,x_i}$ is analogous.

  By invariant~(3), we have that $u,v\in S$ since otherwise $\Phi_{u,v}$ is a part
  of the curve $h$ and therefore does
  not contain any vertices from outside $V(h)$.
  So $\{u,v\}=\{s_{j'},s_{j'+1}\}$ for some $j'\in \{1,\ldots,p-1\}$.
  Let $z\in (V(\pi_i)\cap V(\Phi_{u,v}))\setminus \bnd{P}$.
  If $\Phi_{u,v}=\pi_{s_{j'},s_{j'+1}}$, then the $z\to v$ subpath
  of $\Phi_{u,v}$ and the $x_i\to z$ subpath of $\pi_i$ together form
  an $s_j\to s_{j'+1}$ path in $P$, which by invariant~(2) implies $j'+1\in \{j-1,j,j+1\}$,
  and thus $j'\in [j-2,j]$.
  If $\Phi_{u,v}=\pi_{s_{j'+1},s_{j'}}$, then, analogously, there exists
  an $s_{j}\to s_{j'}$ path in $P$, so $j'\in [j-1,j+1]$.
\end{proof}

By Lemma~\ref{l:atmost5}, for each $\pi_i$, we need to report all vertices
of $\pi_i$ from outside $\bnd{P}$, except those on the union of at most $6$ subpaths
of $\pi_i$. Since the subpaths are always intersections of some two paths
$\pi_{u,v}$, we can identify these subpaths using precomputed data in $O(1)$ time.
Aggregating vertex weights not on at most $6$ subpaths of $\pi_i$ is the same as aggregating
weights on at most $7$ disjoint subpaths of $\pi_i$.
Recall that we have precomputed the aggregate weights for all the subpaths of all the possible $\pi_{u,v}$.
As a result, aggregating or reporting vertices 
$\Pi_P(B)\setminus \bnd{P}$ that lie on any of the paths $\pi_1,\ldots,\pi_l$
takes $O(l)=O(p-g)$ time.
\begin{lemma}
  The general case can be handled in $\Ot(|B|)$ time.
\end{lemma}
\begin{proof}
  Recall that when $T=\emptyset$, we have a base instance that can be solved
  in $O(|B|)$ time.

  Every iteration of the main loop that does not involve changing $B$
  takes $O(\polylog{n})$ time and reduces the size of $T$ by one.
  But the set $T$ can shrink at most $|B|$ times,
  so such iterations cost $\Ot(|B|)$ time in total.
  Every other iteration of the main loop involves reducing the size of $B$
  by some $\Delta>0$ in $\Ot(\Delta)$ time.
  Such iterations clearly cost $\Ot(|B|)$ time in total as well.
\end{proof}

\newcommand{\first}{\textrm{first}}
\newcommand{\last}{\textrm{last}}

\subsection{Handling multiple-hole queries}\label{s:multiple-hole}
In this section we discuss how to drop the simplifying assumption that the query set $A\subseteq \bnd{P}$
lies entirely on a single hole of $P$.

Before we continue, let us remark that in Section~\ref{s:ds-planar} we did not fully leverage
the fact that the holes' bounding cycles of a piece of an $r$-division with few holes
contain vertices from $\bnd{P}$ exclusively.
Indeed, we only used that (i) $A$ lies on a single face $h$ of $P$, and (ii) that there was no
need to report or aggregate vertices lying on the bounding cycle of $h$ -- indeed,
in a piece of an $r$-division with few holes no vertex of $V(P)\setminus \bnd{P}$ lies on the bounding cycle of $h$.

In the following, we will slightly relax the problem: 
define $\bnd^*{P}$ to contain
all the vertices of $P$ on the bounding cycles of $h_1,\ldots,h_k$,
and let $\bnd{P}$ be an arbitrary fixed subset of $\bnd^*{P}$ of size $O(\sqrt{r})$
(but the set $\bnd^*{P}$ can be possibly larger).
Given $A\subseteq \bnd{P}$, the query's objective is to aggregate/report
vertices of $\Pi_P(A)\setminus \bnd^*{P}$.
By the above discussion, the data structure of Section~\ref{s:ds-planar} can
actually solve this slightly more general problem within the same bounds.
Moreover, as discussed in Section~\ref{s:acyclic}, we no longer use the assumption
of Theorem~\ref{t:ds-planar} that $\Pi_P(A)\cap \bnd{P}=A$ which was only required
to reduce to the acyclic case.

Denote by $\bnd_i^*{P}$, for $i=1,\ldots,k$, the vertices of $\bnd^*{P}$ that lie on the hole $h_i$.
Recall that $k=O(1)$.
The sets $\bnd_1^*{P},\ldots,\bnd_k^*{P}$ are not necessarily disjoint.
We now describe how to reduce the problem, roughly speaking, to the case when $P$ has one hole less.
The reduction will require $O(r^3)$ additional preprocessing
and will add only $O(|A|)$ overhead to the query procedure.
The reduction will preserve acyclicity.

\subsubsection{Preprocessing}
Suppose first that there exist no such $s\in \bnd_i^*{P}$ and $t\in \bnd_j^*{P}$, where $i\neq j$,
that $s$ can reach $t$ in~$P$. This can be clearly checked in $O(|P|)=O(r)$ time
by running graph search $k$ times.
In this case, we have that $\Pi_P(A)\setminus \bnd^*{P}$ is a disjoint sum of
$\Pi_P(A\cap \bnd^*_i{P})\setminus \bnd^*_i{P}$. Hence,
all we do is simply apply the preprocessing of Section~\ref{s:ds-planar} $k$ times,
for each hole $h_i$ independently, i.e., we set $\bnd^*{P}:=\bnd^*_i{P}$ and $\bnd{P}:=\bnd{P}\cap \bnd^*_i{P}$
in the $i$-th single-hole data structure. This takes $\Ot(kr^3)=\Ot(r^3)$ time.

Let us now assume that such $s,t$ exist and pick them so that we can find a \emph{simple} $s\to t$ path~$Q$ such
that $V(Q)\cap \bnd^*{P}=\{s,t\}$ (we might possibly have $s=t$). For example, such a path $Q$ can be found in $O(r)$ time
by computing, for each pair of holes of $P$, a shortest path between these two holes,
and taking the shortest out of these paths as $Q$. Wlog. let us renumber the sets $\bnd^*_i{P}$
so that $s\in \bnd^*_1{P}$ and $t\in \bnd^*_2{P}$.

For every $b\in\bnd{P}$, we compute (i) the first vertex $\first_b$ on the path $Q$ that $b$ can reach in $P$,
and (ii) the last vertex $\last_b$ on $Q$ that can reach $b$ in $P$.
All the vertices $\first_b$ and $\last_b$ (that exist) can be computed in $O(r)$ time via graph search (see e.g.~\cite{Thorup04}).

Let us obtain a graph $P'$ from $P$ by cutting along the path $Q=u_1\ldots u_q=s\to t$. Formally, the path $Q$ is
replaced with two copies $Q'=u_1'\ldots u_q'$ and $Q''=u_1''\ldots u_q''$,
but the copies $u_i',u_i''$ of the vertex $u_i$ are not connected. Neither the vertices nor the edges of $Q$
are contained in the graph~$P'$, only $Q'$ and $Q''$ are there.
For all $v\in V(P)\setminus V(Q)$, for convenience we set $v'=v'':=v$.
Clearly, $P'$ is acyclic if $P$ is acyclic.

We set the boundary of $P'$ to be:
\begin{equation*}
  \bnd{P'}:=\bnd{P}\setminus \{s,t\}\cup \bigcup_{b\in \bnd{P}}\{\first_b',\first_b'',\last_b',\last_b''\}.
\end{equation*}

Note that the two holes $h_1,h_2$ of $P$ are merged into a single hole $h$ in $P'$. As a result, and
since all the vertices $\bnd{P'}\setminus \bnd{P}$ lie on $h$,
$P'$ has less holes (that is, faces containing boundary vertices $\bnd{P'}$) than~$P$.
Moreover, we have (recall that $\bnd^*{P'}$ denotes all vertices on the holes of $P'$):
\begin{equation*}
  \bnd^*{P'}:=\bnd^*{P}\setminus \{s,t\}\cup V(Q')\cup V(Q'').
\end{equation*}
Observe that $|P'|\leq 3|P|=O(r)$ and $|\bnd{P'}|\leq 5|\bnd{P}|$ so $|\bnd{P'}|=O(\sqrt{r})$.
We recursively build a data structure supporting desired queries $A'\subseteq \bnd{P'}$ on the graph
$P'$ as defined above.

Finally, for all the subpaths $R$ of $Q$ we precompute the aggregate weights
of all the vertices of $V(R)\setminus \bnd^*{P}=V(R)\setminus \{s,t\}$.
This can be done in $O(r^3)$ time using a brute force algorithm.

Since the piece $P$ has initially size $O(r)$, $|\bnd{P}|=O(\sqrt{r})$ and $\bnd{P}$
lies on $k=O(1)$ holes of~$P$, after at most $k-1$ recursive preprocessing steps
we will arrive at the single-hole case (when a pair $s,t$ cannot be found, possibly because $k=1$).
In each of the recursive data structures, the input piece has size $2^{O(k)}\cdot O(r)=O(r)$,
and has $|\bnd{P}|=2^{O(k)}\cdot O(\sqrt{r})=O(\sqrt{r})$.
At each of the $O(1)$ recursive levels, the additional preprocessing takes $O(r^3)$ time.

\subsubsection{Query}

Given $A\subseteq\bnd{P}$, the query procedure is as follows. 
If $P$ is at the lowest level of the recursive data structure,
then we split $A$ into $O(1)$ sets
$A_i=A\cap \bnd_i^*{P}$.
For each of the sets $A_i$ we run the single-hole query procedure independently.
Since there are no paths between boundary vertices on distinct
holes, we can aggregate the results in an obvious way.
This is justified by the fact that $\Pi_P(A)\setminus \bnd^*{P}$ is a disjoint union
of the $O(1)$ sets $\Pi_P(A\cap \bnd_i^*{P})\setminus \bnd^*{P}$.
Clearly, the query algorithm runs in $\Ot(|A|)$ time in this case.

Otherwise, if $P$ is not leaf-level in the recursion, let $\prec$ denote the natural
order on vertices of the path $Q$.
Let $f=\min_{\prec}\{\first_a:a\in A\}$ and $l=\max_{\prec}\{\last_a:a\in A\}$.
We construct the set
\begin{equation*}
  A':=\begin{cases}
    (A\cap \bnd{P'})\cup \{f',f'',l',l''\}&\text{ if }f\preceq l,\\
    A\cap \bnd{P'}&\text{ otherwise. }
    \end{cases}
\end{equation*}
Clearly, $A'\subseteq \bnd{P'}$ can be constructed in $O(|A|)$ time and has size $O(|A|)$.
Note that $A'$ might not satisfy $\Pi_{P'}(A')\cap \bnd{P'}=A'$, but
we do not need that assumption if $P'$ is acyclic, as discussed before.
We issue a query about $A'$ to the data structure built recursively
on $P'$ which has one hole less.
That data structure produces aggregate weight $\sigma$ of $\Pi_{P'}(A')\setminus \bnd^*{P'}$
or an iterator reporting the elements of that set in $O(1)$ time.
To lift $\sigma$ to the desired aggregate weight for $\Pi_{P}(A)\setminus \bnd^*{P}$,
we simply additionally add the precomputed aggregate weight of the subpath
$Q[f,l]$ (except $s,t\in \bnd^*{P}$).
To lift the reporting function, we simply report the vertices $V(Q[f,l])\setminus \{s,t\}$
after we have finished reporting $\Pi_{P'}(A')\setminus \bnd^*{P'}$.
Recall that $V(Q[f,l])\setminus \{s,t\}$ is disjoint 
with $\Pi_{P'}(A')\setminus \bnd^*{P'}$ by construction.

We now prove the correctness of this algorithm.

\begin{lemma}
  For any $A\subseteq \bnd{P}$, we have $\Pi_{P}(A)\setminus \bnd^*{P}=(\Pi_{P'}(A')\setminus \bnd^*{P'})\cup (V(Q[f,l])\setminus \{s,t\})$.
\end{lemma}
\begin{proof}
  Let us start with the ``$\supseteq$'' direction. Suppose first that $v\in \Pi_{P'}(A')\setminus \bnd^*{P'}$.
  Since we have $V(P')\setminus V(P)\subseteq \bnd^*{P'}$ and $s,t\notin V(P')$, we have $v\in V(P)\setminus \bnd^*{P}$.
  Moreover, $v$ lies on some path $R'=a\to b$ between vertices of $A'$ in $P'$.
  Since every vertex (edge) of $P'$ is a either an original vertex (edge, resp.) or a copy of a vertex (edge, resp.)
  of $P$, there is a path $R=a_0\to b_0$ in $P$ whose $R'$ is a copy of.
  If $a_0,b_0\in A$, then we have $v\in \Pi_P(A)\setminus \bnd^*{P}$ immediately.
  Otherwise, if e.g., $a_0\notin A$, then $f\preceq l$ and $a_0\in \{f,l\}$.
  However, by the definition of $f$, $f$ is reachable from a vertex $c\in A$.
  Since a path $f\to l$ exists in $P$ as a subpath of $Q$, $l$ is reachable from $c$ as well.
  As a result, even if $a_0\notin A$, $a_0$, and thus also $v$, is reachable from $A$ in $P$.
  Analogously one can argue that $b_0$, and thus also $v$, can reach $A$ in $P$.

  Now consider $v\in V(Q[f,l])\setminus \{s,t\}$. Then, $f\preceq l$ and we similarly argue
  that $v$ can be reached from $f$ and can reach $l$ and as a result can reach and can be reached from $A$ in $P$.

  Let us now move to the more challenging ``$\subseteq$'' direction.
  Consider some $v\in \Pi_P(A)\setminus \bnd^*{P}$.
  There exists a path $R=a\to v\to b$ in $P$ with $a,b\in A$.
  If $V(R)\cap V(Q)=\emptyset$, then the path $R$ is preserved in the graph~$P'$.
  Moreover, $a,b\in \bnd{P}\setminus \{s,t\}$ and thus $a,b\in A'$.
  We obtain $v\in \Pi_{P'}(A')\setminus \bnd^*{P'}$ by $v\notin V(Q)$, as desired.
  Hence, in the following we assume $V(R)\cap V(Q)\neq\emptyset$.
  Note that this also implies that $f\preceq l$ since $A$ can reach $Q$ and be reached from $Q$.

  First suppose that $v\in V(Q)\cap V(R)$. Since $v\notin \bnd^*{P}$,
  $v\notin \{s,t\}$. By the definition of $f,l$,
  we have $f\preceq v\preceq l$ and thus $v\in V(Q[f,l])\setminus \{s,t\}$,
  as desired.
  
  Now assume $v\in V(R)\setminus V(Q)$.
  Note that $v\notin \bnd^*{P}\cup V(Q)$ implies $v=v'=v''\notin \bnd^*{P'}$.
  Let $x,y$ denote the first and the last vertex of $Q$ on $R$, respectively.
  Let $R=R_1\cdot R_2\cdot R_3$, where $R_1=a\to x$ and $R_3=y\to b$.
  Observe that the paths $R_1$ and $R_3$ are preserved in $P'$,
  except $R_1$ reaches a copy of $x$, and $R_3$ starts at a copy of $y$.
  Moreover, note that $x\preceq y$ by acyclicity, and thus we have $f\preceq \first_a\preceq x\preceq y\preceq \last_b\preceq l$.
  
  Consider the case $v\in V(R_1)\setminus V(Q)$.
  Assume wlog. that $R_1$ joins the path $Q$ in $P$ in such a way
  that $R_1$ is preserved as $R_1'=a'\to x'$ in $P'$.
  Now, observe that
  the path $R_1'\cdot Q'[x',l']$ is an $a'\to l'$ path in $P'$
  containing the vertex $v$, and $l'\in A'$.
  If $a\notin \{s,t\}$, then clearly also $a'\in A'$.
  Otherwise, if $a=s$, then $f=s$ and thus $a'=f'$, whereas
  if $a=t$, then $l=t$ and thus $a'=l'$. In all cases $\{a',l'\}\subseteq A'$.
  As a result, indeed $v\in \Pi_{P'}(A')\setminus \bnd^*{P'}$.
  One can proceed with the case $v\in V(R_3)\setminus V(Q)$ symmetrically.

  Finally, suppose $v\in V(R_2)$. Then, split $R_2$ into maximal subpaths with only
  their endpoints lying on the path $Q$. One of the obtained maximal subpaths
  $S=c\to d$, where $c,d\in V(Q)$, goes through~$v$.
  Again, we have $f\preceq \first_a\preceq c\preceq d\preceq \last_b\preceq l$.
  Moreover, $S$ is preserved in $P'$ as wlog. the path $S''=c''\to d''$.
  But then $Q''[f'',c'']\cdot S'' \cdot Q''[d'',l'']$
  is a $f''\to l''$ path in $P'$ that goes through $v$.
  Since $f'',l''\in A'$, we again obtain $v\in \Pi_{P'}(A')\setminus \bnd^*{P'}$, as desired.
\end{proof}

When queried about the set $A$, the recursive data structures are passed
a set $A'$ with $|A'|=|A|+O(1)=O(|A|)$. As a result, the query time is dominated
by the $\Ot(|A'|)=\Ot(|A|)$ time used by the bottommost single-hole data structure.
This finishes the proof of Theorem~\ref{t:ds-planar}.

\subsection{The existential reachability data structure}\label{a:anyreachproof}

  In this section, we prove the following.
\lanyreach*

  We will use the following property of reachability between vertices
  lying on a single face of a planar digraph $P$
  essentially proved in~\cite{ItalianoKLS17}.

\begin{restatable}{lemma}{reachpartition}\label{l:reach-partition}{\normalfont \cite{ItalianoKLS17}}
  Let the vertices $Y$ lie on a single face $h$ of a planar digraph $P$. Then in $\Ot(|P|+|Y|^2)$ time
  one can compute $O(|Y|)$ pairs of subsets $S_i,T_i\subseteq Y$ such that:
  \begin{enumerate}
    \item For any $i$, the elements $T_i$ have an associated linear order $\prec_i$.
    \item For any $i$ and $s\in S_i$, there exist such vertices $l_{i,s},r_{i,s}$
      that for any $t\in T_i$, $s$ can reach $t$ in $P$ iff $l_{i,s}\preceq_i t\preceq_i r_{i,s}$.
    \item Every $v\in Y$ is contained in $O(\log{|Y|})$ subsets $S_i$ and $O(\log{|Y|})$ subsets $T_i$.
    \item For any $u,v\in Y$, if $u$ can reach $v$ in $P$, then there exists such $i$ that $(u,v)\in S_i\times T_i$.
  \end{enumerate}
\end{restatable}
\begin{proof}[Proof sketch.]
Split $Y$ into two disjoint parts $A,B$ such that $A\cup B=Y$, $|A|,|B|\leq \left\lceil|Y|/2\right\rceil$,
and both $A,B$ appear consecutively on $h$.
Let $T$ denote the vertices of $B$ that can be reached from some vertex of $A$ in $P$.
Similarly, let $S$ denote the vertices of $A$ that can reach $B$ in $P$.
Let $\prec$ be the counterclockwise order on $T$ starting with the vertex that appears
earliest after a vertex of $A$ in the counterclockwise order of $Y$ on $h$.
  One can easily prove~(see~\cite{ItalianoKLS17}) that each $s\in S$ can in fact reach
some number of vertices that appear consecutively in $T$.
As a result, there exist such $l_s,r_s\in T$ that for any $t\in T$, $s$ can reach $t$ in $P$
iff $l_s\preceq t\preceq r_s$.

We include the pair $(S,T)$ along with the order $\prec$, and the vertices $\{l_s,r_s:s\in S\}$
in the constructed family. The obtained pair encodes the reachability from the vertices $A$ to vertices $B$ in $P$:
if $a\in A$ can reach $b\in B$ in $T$ then clearly $a\in S$, $b\in T$.
Note that the subsets $S,T$ can be constructed in $\Ot(|Y|^2)$ time if only
  we can answer reachability queries in $P$ in $O(\polylog{n})$ time.
This is indeed possible after $\Ot(|P|)$ time preprocessing~\cite{Thorup04}.

Symmetrically we can deal with the reachability from the vertices $B$ to the vertices $A$ in $P$.
It remains to handle reachability between pairs of vertices inside the set $A$, and between pairs inside~$B$.
This is done recursively. Since $|A|,|B|\leq \lceil|Y|/2\rceil$, each vertex of the original
  input set~$Y$ will appear in $O(\log|Y|)$ recursive calls,
  and thus in $O(\log|Y|)$ sets $S_i$ or $T_i$.
  It is easy to check that the algorithm produces all the pairs $(S_i,T_i)$
  in $\Ot(|Y|^2)$ time.
\end{proof}

\begin{proof}[Proof of Lemma~\ref{l:any-reach}]
  We first compute the sets $(S_i,T_i)$ from Lemma~\ref{l:reach-partition} applied to $Y:=\bnd{P}\cap V(h)$ (along with the orderings $\prec_i$,
  and all the vertices $l_{i,v},r_{i,v}$) in $\Ot(|P|+|\bnd{P}|^2)=\Ot(r)$ time.
  The data structure simply stores
  the sets $T_i\cap Z$ for all $i$ sorted according to $\prec_i$, each of them in a balanced binary search tree.
  Note that when an element $b$ is inserted or removed from $Z$, it is enough
  to update $O(\log{r})$ sets $T_i\cap Z$ as $b$ is contained in $O(\log{r})$ sets $T_i$.

  To answer a query, that is, find a vertex $z\in Z$ reachable from a query vertex $v$,
  for each of the $O(\log{r})$ indices $i$ such that $v\in S_i$ we check whether
  the set $T_i\cap Z$ contains an element $z$ between $l_{i,s}$ and $r_{i,s}$ inclusive, wrt. $\prec_i$.
  This clearly takes $O(\log{r})$ time.
  By Lemma~\ref{l:reach-partition}, $s$ can reach a vertex in $Z$ if and only if
  at least one $z$ is found this way.

  We conclude that both the update and the query time of the data structure is $O(\log^2{r})$.
\end{proof}

\section{Fully dynamic \#SSR in planar digraphs}\label{s:ssr}
In this section, we sketch how a data structure summarized by the below lemma can be obtained.
\begin{restatable}{lemma}{tplanarssrcount}\label{t:planar-ssr-count}
    Let $G$ be a planar digraph subject to planarity-preserving edge insertions and deletions. There exists a data structure with
$\Ot(n^{4/5})$ worst-case update time
    supporting counting the vertices
    reachable from a query vertex $s$ in $\Ot(n^{4/5})$ time.
  \end{restatable}
We again use the same template as in Sections~\ref{s:planar-conn}~and~\ref{s:planar-sccs}:
as a base, we maintain an $r$-division $\rdiv$ with few holes, along with the piecewise reachability
certificates $X_P$. 

In addition, we leverage the additively-weighted Voronoi diagram data structure of~\cite{GawrychowskiKMS21}.
Before stating the interface and characteristics of that data structure, we need to introduce some
more notation. Suppose $H$ is a non-negatively weighted plane digraph.
Let ${\bnd{H}\subseteq V(H)}$ be the set of vertices of $H$ lying on some $O(1)$ fixed faces of $H$ (and containing
all vertices of these faces) and such that every $v\in V(H)$ is reachable from $\bnd{H}$ in $H$.
Let $\omega:\bnd{H}\to\mathbb{R}$ be a \emph{weight function} such that for any $v\in V(H)$, $\omega(s)+\dist_H(s,v)$ has
a unique minimizer $s\in \bnd{H}$.
The \emph{additively weighted Voronoi diagram} wrt. $\omega$ is the partition of $V(H)$
into disjoint subsets (cells) $\Vor(s)$, $s\in \bnd{H}$, such that
$\Vor(s)$ contains precisely the vertices $v$ such that $s$ minimizes $\omega(s)+\dist_H(s,v)$.
Note that by our assumptions on $H$, if $v\in \Vor(s)$, then $\omega(s)+\dist_H(s,v)\neq\infty$.
The complexity $|\bnd{\Vor(s)}|$ of $\Vor(s)$ is defined to be the number
of faces of $H$ whose vertices include elements of $\Vor(s)$ and
at least two other cells. One can prove that $\sum_{s\in \bnd{H}}|\bnd{\Vor(s)}|=O(|\bnd{H}|)$~\cite{GawrychowskiKMS21}.

\begin{theorem}\label{t:voronoi}{\upshape\cite{GawrychowskiKMS21}}
  Let $H$ be as defined above. In $\Ot\left(|H|\cdot |\bnd{H}|^2\right)$ one can construct a data structure
  such that for any $\omega$ satisfying the assumptions outlined above, it can build a representation
  of the additively-weighted Voronoi diagram wrt. $\omega$ in $\Ot(|\bnd{H}|)$ time.
  The representation allows, among others, computing all the cell sizes {\upshape $|\Vor(s)|$}, $s\in \bnd{H}$, in $\Ot(|\bnd{H}|)$ total time.
\end{theorem}
\begin{proof}[Proof sketch.]
  This is essentially a special case of~\cite[Theorem~1.1]{GawrychowskiKMS21}. There, it is stated
  that after the representation is built, one can compute, for any $s\in \bnd{H}$, the furthest vertex 
  $v\in \Vor(s)$ (i.e.,~$v$ maximizing $\omega(s)+\dist_H(s,v)$) in $\Ot(|\bnd{\Vor}(s)|)$ time.
  But in~\cite[Section~7]{GawrychowskiKMS21}, where this functionality is described,
  it is evident that within the same time bound one can in fact aggregate over $\Vor(s)$ (that is, compute
  the maximum or e.g., sum) arbitrary vertex labels $\ell(v)$ (which could even dependent on $s$; as long as they are computable in $\Ot(|H|\cdot|\bnd{H}|^2)$
  total additional time).
  Therefore, if we assign the label $1$ to every vertex of $H$, in $\Ot(|\bnd{\Vor}(s)|)$ time
  one can compute $|\Vor(s)|$. Through all $s\in |\bnd{H}|$, this is $\Ot(|\bnd{H}|)$ time as mentioned before.
\end{proof}
For each piece $P$, besides the certificate $X_P$,
we additionally maintain a data structure $\mathcal{AV}_P$ of Theorem~\ref{t:voronoi} for a ``proxy'' $P'$ obtained from $P$ by removing all the vertices $v\in V(P)$ unreachable from $\bnd{P}$. Observe that $\mathcal{AV}_P$ can be computed in $\Ot(|P|\cdot (|\bnd{P}|)^2)=\Ot(r^2)$ time.
By Theorem~\ref{t:dyn-rdiv}, updating $\rdiv$ along with all other associated piecewise data structures
costs $\Ot(r^2)$ worst-case time.

Upon query about the number of vertices reachable from $s\in V$, we proceed as follows.
We first compute the subset $S\subseteq \bnd{\rdiv}$ of boundary vertices reachable
from $s$ in $G$ by running graph search from $s$ on the graph $P_s\cup X=P_s\cup\bigcup_{P\in\rdiv}X_P$,
where $P_s$ is an arbitrarily chosen piece with containing $s$.
This is correct by Lemma~\ref{l:bnd-cert}
and costs $\Ot(r+n/\sqrt{r})$ time. 
From the graph search we also retrieve the vertices $T\subseteq V(P_s)\setminus \bnd{P_s}$ that $s$ can reach.
Then, for each $P\neq P_s$ we construct a weight function $\omega_P$
such that if $\bnd{P}\cap S=\{a_1,\ldots,a_k\}$ and $\bnd{P}\setminus S=\{b_1,\ldots,b_l\}$, then:
$\omega(a_i):=i\cdot M$ for all $i$ and $\omega(b_j):=(k+j)\cdot M$ for all $j$, where $M=V(P)$. Note that since
$\dist_{P'}(\bnd{P},v)<M$ for any $v\in V(P')$, the cell is uniquely determined for every $v\in V(P')$.
Namely, if $v$ can be reached from $a_i\in \bnd{P}\cap S$ in $P'$, and $i$ is minimized (and thus $s$ can reach $v$ through $S$), then $v\in\Vor(a_i)$.
Otherwise, $v$ can be reached from some $b_j\in\bnd{P}\setminus S$ in $P$, where $j$ is minimized (and $s$ cannot reach $v$ through $S$), and then $v\in \Vor(b_j)$.
We conclude that $\bigcup_{i=1}^k \Vor(a_i)$ contains precisely the vertices
$S_P\subseteq V(P)$ (including $\bnd{P}\cap S$) that $s$ can reach via paths through $\bnd{P}$. By Theorem~\ref{t:voronoi},
their number  $|S_P|=\sum_{i=1}^k |\Vor(a_i)|$ can be computed in $\Ot(\sqrt{r})$ time.
In fact, these are \emph{all} vertices that~$s$ can reach in $P$:
since $P\neq P_s$, either $s\notin V(P)$ or $s\in \bnd{\rdiv}$, so every vertex in $P$
can be reached from $s$ only via some boundary vertex of $P$ anyway.
We conclude that we can count vertices reachable from $s$ in $G$ by computing:
\begin{equation*}
  |S|+|T|+\sum_{P_s\neq P\in\rdiv}\left(|S_P|-|S\cap \bnd{P}|\right).
\end{equation*}
The query time is thus $\Ot(r+n/\sqrt{r})$. By picking $r=n^{2/5}$, we make both the query time and the worst-case update time $\Ot(n^{4/5})$.

\bibliographystyle{alpha}
\bibliography{references}

\end{document}